\documentclass[11pt]{article}

\usepackage{amssymb, amsthm,amsmath, nicefrac}
\usepackage{graphicx, fullpage}
\usepackage[margin=2.5cm]{geometry}

\newcommand {\rounddown} [1] {{\lfloor {#1} \rfloor}}

\newcommand {\Exp} {\mathbb{E}}

\newcommand {\E} [1] {\Exp\left[#1\right]}

\newcommand{\given}{\,\mid\,}
\newcommand{\Given}{\;\;\mid\;\;}

\newcommand {\cx} {\leq_{cx}}

\newcommand {\fup} {\varphi^{\text{\tiny{$\mathbf{\uparrow}$}}}}
\newcommand {\fdown} {\varphi^{\text{\tiny{$\mathbf{\downarrow}$}}}}

\newcommand {\bbR} {\mathbb{R}}
\newcommand {\bbN} {\mathbb{N}}

\newcommand {\calI} {{\cal{I}}}
\newcommand {\calY} {{\cal{Y}}}

\newcommand {\calD} {{\cal{D}}}

\newcommand {\calP} {{\cal{P}}}

\newtheorem{theorem}{Theorem}[section]

\newtheorem{corollary}[theorem]{Corollary}

\newtheorem{lemma}[theorem]{Lemma}
\newtheorem{fact}[theorem]{Fact}
\newtheorem{claim}[theorem]{Claim}
\newtheorem{definition}[theorem]{Definition}
\newtheorem{remark}[theorem]{Remark}

\title{Solving Optimization Problems with Diseconomies of Scale via Decoupling\footnote{
The conference version of this paper appeared at FOCS 2014.}}
\author{Konstantin Makarychev\\Microsoft Research \and Maxim Sviridenko\\Yahoo! Labs}
\date{}

\begin{document}

\maketitle
\begin{abstract}
We present a new framework for solving optimization problems with a diseconomy of scale. In such problems, our goal is to minimize the cost of resources used to perform a certain task. The cost of resources grows superlinearly, as $x^q$, $q\ge 1$, with the amount $x$ of resources used. We define a novel linear programming relaxation for such problems, and then show that the integrality gap of the relaxation is $A_q$, where $A_q$ is the $q$-th moment of the Poisson random variable with parameter 1. Using our framework, we obtain approximation algorithms for the Minimum Energy Efficient Routing, Minimum Degree Balanced Spanning Tree, Load Balancing on Unrelated Parallel Machines, and Unrelated Parallel Machine Scheduling with Nonlinear Functions of Completion Times problems.

Our analysis relies on the decoupling inequality for nonnegative random variables. The
inequality states that
$$\big \|\sum_{i=1}^n X_i\big\|_{q} \leq C_q \,\big \|\sum_{i=1}^n Y_i\big\|_{q},$$
where $X_i$ are independent nonnegative random variables, $Y_i$ are possibly dependent
nonnegative random variable, and each $Y_i$ has the same distribution as $X_i$.
The inequality was proved by de la Pe\~na in 1990. De la Pe\~na,
Ibragimov, and Sharakhmetov showed that $C_q\leq 2$ for $q\in (1,2)$ and
$C_q\leq A_q^{\nicefrac{1}{q}}$ for $q\geq 2$. We show that the optimal constant is
$C_q=A_q^{\nicefrac{1}{q}}$ for any $q\geq 1$. We then prove a more general inequality:
For every convex function $\varphi$,
$$\E{\varphi\Big(\sum_{i=1}^n X_i\Big)} \leq \E{\varphi\Big(P\sum_{i=1}^n Y_i\Big)},$$
and, for every \emph{concave} function $\psi$,
$$\E{\psi\Big(\sum_{i=1}^n X_i\Big)} \geq \E{\psi\Big(P\sum_{i=1}^n Y_i\Big)},$$
where $P$ is a Poisson random variable with parameter 1 independent of the random variables $Y_i$.
\end{abstract}

\setcounter{page}{0}
\thispagestyle{empty} 
\pagebreak

\section{Introduction}
In this paper, we study combinatorial optimization problems with a diseconomy of scale. We consider problems in which we need to
minimize the cost of resources used to accomplish a certain task. Often, the cost grows linearly with the amount of resources used.
In some applications, the cost is sublinear e.g., if we can get a discount when we buy resources in bulk. Such phenomenon is known as ``economy of scale". However, in many applications the cost is superlinear.
In such cases, we say that the cost function exhibits a ``diseconomy of scale''. A good example of a diseconomy of scale is the cost of energy used for computing.
Modern hardware can run at different processing speeds. As we increase the speed, the energy consumption grows superlinearly. It
can be modeled as a function $P(s) = c s^{q}$ of the processing speed $s$, where $c$ and $q$ are parameters that depend on the specific hardware. Typically, $q\in (1,3]$ (see e.g., \cite{A,IK,power-ref}).

As a running example, consider the Minimum Power Routing problem studied by Andrews, Fern\'andez Anta, Zhang, and Zhao~\cite{AAZZ}. We are given a graph $G=(V,E)$
and a set of demands $\calD=\{(d_i, s_i, t_i)\}$. Our goal is to route $d_i$ ($d_i\in \bbN$) units of demand $i$ from the source $s_i\in V$ to the
destination $t_i\in V$ such that every demand $i$ is routed along a single path $p_i$ (i.e. we need to find an unsplittable multi-commodity flow).
We want to minimize the energy cost. Every link (edge) $e\in E$ uses $f_e(x_e)=c_e x_e^{q}$ units of power, where
$c_e$ is a scaling parameter depending on the link $e$, and $x_e$ is the load on $e$.

The straightforward approach to solving this problem is as follows. We define a mathematical programming relaxation that routes demands fractionally. It sends
$y_{i,p}d_{i}$ units of demand via the path $p$ connecting $s_i$ to $t_i$. We require that $\sum_p y_{i,p} = 1$ for every demand $i$.
The objective function is to minimize
$$\min \sum_{e\in E} c_e x_e^{q} = \min \sum_{e\in E} c_e \big(\sum_{p:e\in p} y_{i,p} d_i \big)^{q},$$
where $x_e = \sum_{p:e\in p} y_{i,p} d_i$ is the load on the link $e$. This relaxation can be solved in polynomial time, since the objective function
is convex (for $q \geq 1$). But,
unfortunately, the integrality gap of this relaxation is $\Omega(n^{q-1})$ \cite{AAZZ}. Andrews et al.~\cite{AAZZ} gave the
following integrality gap example. Consider two vertices $s$ and $t$ connected via $n$ disjoint paths. Our goal is to route 1 unit of
flow integrally from $s$ to $t$. The optimal solution pays 1. The LP may cheat by routing $1/n$ units of flow via $n$ disjoint paths. Then,
it pays only $n\times (1/n)^{q} = n^{1-q}$.

For the case of uniform demands, i.e., for the case when all $d_i=d$,
Andrews et al.~\cite{AAZZ} suggested a different objective function:
$$\min \sum_{e\in E} c_e \max\{x_e^{q}, d^{q - 1} x_e\}.$$
The objective function is valid, because in the integral case, $x_e$ must be a multiple of $d$, and thus $x_e^{q}\geq d^{q - 1} x_e$.
Andrews et al.~\cite{AAZZ} proved that the integrality gap of this relaxation is a constant. Bampis et al.~\cite{BKLLS} improved the bound to the \emph{fractional Bell number}
$A_{q}$ that is defined as follows: $A_{q}$ is the $q$-th moment of the Poisson random variable $P_1$
with parameter 1 (see Figure~\ref{fig:plot-Aq} in Appendix~\ref{sec:figures}). I.e.,
\begin{equation}\label{def:frac-Bell-number}
A_{q} = \E {P_1^{q}}=\sum_{t=1}^{+\infty}t^q \frac{e^{-1}}{t!}.
\end{equation}

For the case of general demands no constant approximation was known. The best known approximation due to Andrews et al.~\cite{AAZZ} was
$O(k+\log^{q-1} \Delta)$ where $k=|\calD|$ is the number of demands and $\Delta=\max_{i} d_i$ is the size of the largest
demand (Theorem 8 in \cite{AAZZ}).

In this work, we give an $A_{q}$-approximation algorithm for the general case and thus close the gap between the case of uniform
and non-uniform demands. Our approximation algorithm uses a general framework for solving problems with a diseconomy of scale which we present in this paper. We
use this framework to obtain approximation algorithms for several other combinatorial optimization problems. We give $A_{q}^{\nicefrac{1}{q}}$-approximation algorithm for Load Balancing on Unrelated Parallel Machines (see Section~\ref{section:LoadBalancing}), $2^qA_{q}$-approximation algorithm for
Unrelated Parallel Machine Scheduling with Non-linear Functions of
Completion Times (see Section~\ref{section:UnrelatedParallel}) and $A_{q}$-approximation algorithm for the Minimum Degree Balanced Spanning Tree problem (see Section~\ref{section:QST}). The best previously known bound for the first problem with $q\in[1,2]$
was $2^{1/q}$ (see Figure~\ref{fig:plot-results} for comparison). The bound is due to Kumar, Marathe, Parthasarathy and Srinivasan~\cite{KMPS}.
There were no known approximation guarantees for the latter problems.

In the analysis, we use the de la Pe\~na decoupling inequality~\cite{Pena90, Pena99}.
\begin{theorem}[de la Pe\~na~\cite{Pena90, Pena99}]
Let $Y_1,\dots, Y_n$ be jointly distributed nonnegative (non-independent) random variables, and
let $X_1,\dots, X_n$ be independent random variables such that
each $X_i$ has the same distribution as $Y_i$. Then, for every $q \geq 1$,
\begin{equation}\label{eq:de-la-Pena}
\big \|\sum_{i=1}^n X_i\big\|_{q} \leq C_q \,\big \|\sum_{i=1}^n Y_i\big\|_{q},
\end{equation}
for some universal constant $C_q$.
\end{theorem}
De la Pe\~na, Ibragimov, and Sharakhmetov (\cite{PIS}, Corollary~3.4)
showed that $C_q \leq 2$ for $q\in [1,2]$, and
$C_q \leq A_q^{1/q}$ for $q\geq 2$. We give an alternative
proof of this inequality, and show that the inequality holds for $C_{q}=A_{q}^{1/q}$ for any $q\geq 1$, and moreover this bound is tight (for any $q\geq 1$). Thus, we improve the known upper bound for $C_q$ for $q\in (1,2)$.
\begin{theorem}\label{thm:decoupling}
Inequality~(\ref{eq:de-la-Pena}) holds for $C_{q}= A_{q}^{\nicefrac{1}{q}}$, where $A_{q}$ is the fractional Bell number (see~Equation~(\ref{def:frac-Bell-number}) and Figure~\ref{fig:plot-Aq}).
Moreover, $A_{q}^{\nicefrac{1}{q}}$ is the optimal upper bound on $C_{q}$.
\end{theorem}
In fact we prove a more general inequality for arbitrary convex functions and an analogous inequality for concave functions.
In Section~\ref{sec:neg-depend} (see Corollary~\ref{cor:neg-assoc}), we extend this theorem to negatively associated
random variables~$X_i$.

\subsection{General Framework}
We now describe the general framework for solving problems with a diseconomy of scale. We consider optimization problems with $n$ decision
variables $y_1,\dots,y_n\in\{0, 1\}$. We assume that the objective function equals the sum of $k$ terms, where the $j$-th term is of the form
$$f_j\Big(\sum_{i=1}^n d_{ij} y_{i}\Big),$$
here $f_j$'s are nonnegative monotonically nondecreasing convex functions, $d_{ij}\geq 0$ are parameters. The vector $y=(y_1,\dots,y_n)$ must satisfy the constraint $y\in \calP$
for some polytope $\calP\subset [0,1]^n$. Therefore, the optimization problem can be written as the following boolean convex program (IP):
\begin{eqnarray}
\min &&\sum_{j\in [k]} f_j \Big(\sum_{i\in [n]} d_{ij} y_{i}\Big) \label{IP:obj}\\
&&y\in {\cal P} \label{IP:poly}\\
&&y\in \{0,1\}^{n} \label{IP:int}
\end{eqnarray}
We assume that we can optimize any linear function over the polytope ${\cal P}$ in polynomial time (e.g., $\cal P$ is defined by polynomially many
linear inequalities, or there exists a separation oracle for $\cal P$). Thus, if we replace the integrality constraint (\ref{IP:int}) with the
relaxed constraint $y\in [0, 1]^n$ (which is redundant, since $\calP\subset [0,1]^n$), we will get a convex programming problem that can be solved in polynomial time (see \cite{BV}). However, as we have seen in the example of Minimum Power Routing, the integrality gap of the relaxation can be as large as $\Omega(n^{q-1})$
for $f_j(t)=t^q$.

In this work, we introduce a linear programming relaxation of (\ref{IP:obj})-(\ref{IP:int}) that has an
integrality gap of $A_q$ for $f_j(t) = c_j t^q$ under certain assumptions on the polytope $\cal P$.
We define auxiliary variables $z_{jS}$ for all $S\subset [n]$ and $j\in [k]$. In the integral solution,
$z_{jS}=1$ if and only if $y_i= 1$ for $i\in S$ and $y_i= 0$ for $i\notin S$.
\begin{equation}
\min \sum_{j\in [k]} \sum_{S\subseteq [n]} f_j \left( \sum_{i\in S} d_{ij}\right)z_{jS} \label{R:obj}
\end{equation}
\begin{align}
y&\in {\cal P},\;\;\;\label{R:poly}\\
\sum_{S\subseteq [n]}z_{jS}&=1, & \forall j\in [k] \label{const1}\\
\sum_{S:i\in S}z_{jS}&=y_i,&\forall i\in [n], j\in [k] \label{const2}&{}\\
z_{jS}&\ge 0,& \forall S\subseteq [n], j\in [k] \label{R:rel3}
\end{align}
\begin{remark}
In the integral solution, $z_{j'S} = z_{j''S}$ for all $j'$ and $j''$. The reason why we introduced
many copies of the same integral variable $z_S$ to the LP is that the LP above is easier to solve than the LP with an extra constraint $z_{j'S} = z_{j''S}$.
\end{remark}

Optimization problem (\ref{R:obj})-(\ref{R:rel3}) is a relaxation of the original problem
(\ref{IP:obj})-(\ref{IP:int}). The LP has exponentially many variables.
We show, however, that the optimal solution to this LP can be found in polynomial time
up to an arbitrary accuracy $(1+\varepsilon)$.

We shall assume that all $d_{ij}$ are integral and polynomially bounded, or, more generally, that all $d_{ij}$ are multiples of
some $\delta_j$ and that $d_{ij}/\delta_j$ are polynomially bounded in $n$ and $1/\varepsilon$. Given an arbitrary instance of the problem, it is easy to round all $d_{ij}$'s to
multiples of a sufficiently small $\delta_j$, so that the cost of any solution changes by at most $(1+\varepsilon)$
assuming that functions $f_j$ satisfy some mild conditions. In Section~\ref{sec:discr} (Theorem~\ref{thm:discr}), we show how to pick $\delta_j$ for
functions $f_j$ satisfying the following conditions: 
\begin{enumerate}
\item Each $f_j$ is a convex increasing function; 
\item For each $j$, $f_j(0)=0$;
\item For each $j$ and $t\in(0,\sum_i d_{ij}]$,
$t(\log f_j(t))'\leq p(n)$, where $p$ is some polynomial. 
\end{enumerate}
Note, that for $f(t) = c t^{q}$, we have
$t(\log f(t))' = t(\log c t^q)' = t\cdot (q/t) = q$. So functions $f_j=c_j t^{q_j}$ satisfy the conditions of Theorem~\ref{thm:discr} if $q_j$ are polynomially bounded.

\begin{theorem}\label{thm:efficient}
Suppose that there exists a polynomial time separation oracle for the polytope ${\cal P}$, $f_j(t)$
is computable in polynomial time as a function of $j$ and $t$, and all $d_{ij}$ are multiples of $\delta_j$ such that $d_{ij}/\delta_j$ are polynomially
bounded   in $n$ and $1/\varepsilon$. Then, for every
$\varepsilon > 0$, there exists a polynomial time algorithm that finds a $(1+\varepsilon)$-approximately
optimal solution to LP (\ref{R:obj})-(\ref{R:rel3}).
\end{theorem}

For a convex non-decreasing function $f$ define $A(f)$ as follows:
\begin{equation}\label{eq:def-general-A}
A(f) = \sup_{t>0} \Exp\Big[\frac{f(tP)}{f(t)}\Big],
\end{equation}
where $P$ is a Poisson random variable with parameter 1. Note that
$$A(c t^q) = \sup_{t>0}\Exp\Big[\frac{c (tP)^q}{ct^q}\Big] =\Exp[P^q] = A_q.$$

We prove the following theorem.

\begin{theorem}\label{thm:main}
Let $D_j= \{i: d_{ij} \neq 0\}$. Assume that there exists a randomized algorithm $R$ that given a $y\in \calP$, returns a random integral point $R(y)$ in $\calP\cap\{0,1\}^n$ such that
\begin{enumerate}
\item $\Pr(R_i(y) = 1) = y_i$ for all $i$ (where $R_i(y)$ is the $i$-th coordinate of $R(y)$);
\item Random variables $\{R_i(y)\}_{i\in D_j}$ are independent or negatively associated (see Section~\ref{sec:neg-depend} for the definition) for every $j$.
\end{enumerate}
Then, for every feasible solution $(y^*,z^*)$ to LP (\ref{R:obj})-(\ref{R:rel3}), we have
\begin{equation}\label{eq:thm:main}
\Exp\Big[\sum_{j\in [k]} f_j \Big(\sum_{i\in [n]} d_{ij} R_i(y^*)\Big) \Big]\leq
\sum_{j\in [k]} A(f_j) \sum_{S\subseteq [n]} f_j\Big( \sum_{i\in S} d_{ij}\Big)z^*_{jS},
\end{equation}
where $A(f_j)$ is defined as in (\ref{eq:def-general-A}).
Particularly, since LP (\ref{R:obj})-(\ref{R:rel3}) is a relaxation for IP (\ref{IP:obj})-(\ref{IP:int}),
if $(y^*,z^*)$ is a $(1+\varepsilon)$-approximately optimal solution to LP (\ref{R:obj})-(\ref{R:rel3}), then
$$\Exp\Big[\sum_{j\in [k]} f_j \Big(\sum_{i\in [n]} d_{ij} R_i(y^*)\Big) \Big]\leq (1+\varepsilon) \max_j(A(f_j)) \,IP,$$
where $IP$ is the optimal cost of the boolean convex program (\ref{IP:obj})-(\ref{IP:int}).
\end{theorem}
\noindent This theorem guarantees that an algorithm $R$ satisfying conditions (1) and (2) has an approximation
ratio of $(1+\varepsilon) A_q$ for $f_j(t) = c_j t^q$.

In the next section, Section~\ref{sec:applications}, we show how to use the framework to obtain 
approximation algorithms for four different combinatorial optimization problems. Then, in Section~\ref{sec:thm:efficient},
we give an efficient algorithm for solving LP (\ref{R:obj})-(\ref{R:rel3}). In Section~\ref{sec:thm:main},
we prove the main theorem -- Theorem~\ref{thm:main}. The proof easily follows from the decoupling inequality,
which we prove in Section~\ref{sec:thm:decoupling}. Finally, in Section~\ref{sec:Generalizations},
we describe some generalizations of our framework.

\section{Applications}\label{sec:applications}
In this section, we show applications of our general technique. We start with the problem discussed in the introduction -- Energy Efficient Routing.
Recall, that Andrews et al. \cite{AAZZ} gave an $O(k+\log^{q-1} \Delta)$-approximation algorithm for this problem where $k=|\cal D|$ and
$\Delta=\max_{i\in\mathcal{D}} d_i$ (Theorem 8 in \cite{AAZZ}). We give an $(1+\varepsilon)A_{q}$-approximation algorithm for any fixed $\varepsilon>0$.

\subsection{Energy Efficient Routing}\label{section:routing}
We write a standard integer program. Each variable $y_{i,e}\in \{0,1\}$ indicates whether the edge $e$ is used to route the flow from $s_i$ to $t_i$.
Below, $\Gamma^+(u)$ denotes the set of edges outgoing from $u$; $\Gamma^-(u)$ denotes the set of edges incoming to $u$.
\begin{equation}
\min \sum_{e\in E} c_e \left(\sum_{i\in \mathcal{D}} d_i y_{i,e}\right)^{q_e} \label{Energy_obj}
\end{equation}
\begin{align}
\sum_{e\in\Gamma^+(u)} y_{i,e} &= \sum_{e\in\Gamma^-(u)} y_{i,e},  &  \forall i,\;  u\in V \setminus \{s_i,t_i\} \label{eqn:r2}\\
\sum_{e\in\Gamma^+(s_i)}y_{i,e}&=1, &  \forall i\; \label{eqn:r3}\\
\sum_{e\in\Gamma^-(t_i)}y_{i,e}&=1, &  \forall i\; \label{eqn:r4}\\
y_{i,e}&\in \{0,1\}, &\forall i,\;e\in E &\label{eqn:r5}
\end{align}

Using Theorem~\ref{thm:efficient}, we obtain an almost optimal fractional solution $(y,z)$ of
LP relaxation (\ref{R:obj})-(\ref{R:rel3})
of IP (\ref{Energy_obj})-(\ref{eqn:r5}). We apply randomized rounding in order to select a path for each demand.
Specifically, for each demand $i\in\mathcal{D}$, we consider the standard flow decomposition into paths: In the decomposition,
each path $p$ connecting $s_i$ to $t_i$ has a weight $\lambda_{i,p}\in\bbR^+$. For every edge $e$,
$\sum_{p:e\in p} \lambda_{i,p} = d_i y_{i,e}$; and $\sum_{p} \lambda_{i,p} = d_i$. For each $i$, the approximation algorithm
picks one path $p$ connecting $s_i$ to $t_i$ at random with probability $\lambda_{i,p}/d_i$, and routes all demands from $s_i$ to $p_i$
via $p$. Thus, the algorithm always obtains a feasible solution.

We verify that the integral solution corresponding to this combinatorial solution satisfies the conditions of Theorem~\ref{thm:main}.
Let $R_{i,e}(y)$ be the integral solution, i.e., let $R_{i,e}(y) = 1$ if the edge $e$ is chosen in the path connecting $s_i$ and $t_i$.
First, $R_{i,e}(y) = 1$ if the path connecting $s_i$ and $t_i$ contains $e$, thus
$$\Pr(R_{i,e}(y) = 1) = \sum_{p: e\in p} \lambda_{i,p}/d_i = y_{i,e}.$$
Second, the paths for all demands are chosen independently. Each $R_{i,e}(y)$ depends only on paths that connect $s_i$ to $t_i$. Thus all random variables
$R_{i,e}(y)$ (for a fixed $e$) are independent. Therefore, by Theorem~\ref{thm:main}, the cost of the solution obtained by the algorithm is bounded by
$(1+\varepsilon)A_{q}\,OPT$, where $OPT$ is the cost of the optimal solution to the integer program which is exactly equivalent to the
Minimum Energy Efficient Routing problem.

\subsection{Load Balancing on Unrelated Parallel Machines}
\label{section:LoadBalancing}
We are given $n$ jobs and $m$ machines. The processing time of the job $j\in [n]$ assigned to the machine $i\in [m]$ is $p_{ij}\ge 0$. The goal is to assign jobs to machines to minimize the $\ell_q$-norm of machines loads. Formally, we partition the set of jobs into $m$ sets $S_1,\dots, S_m$ to minimize $\left(\sum_{i\in [m]} (\sum_{j\in S_i} p_{ij})^q\right)^{1/q}$. This is a classical scheduling problem which is used to model load balancing in
practice\footnote{A slight modification of the problem, where the objective is $\min \sum_{i\in [m]} (\sum_{j\in S_i} p_{ij})^q$, can be used for energy efficient scheduling. Imaging that we need to assign $n$ jobs to $m$ processors/cores so that all jobs are completed by a certain deadline $D$.
We can run processors at different speeds $s_i$. To meet the deadlines we must set
$s_i = D^{-1} \sum_{j\in S_i} p_{ij}$. The total power consumption is proportional to
$D\times \sum_{i=1}^m s_i^q = D^{1-q} \times \sum_{i=1}^m \big(\sum_{j\in S_i} p_{ij}\big)^q$.
For this problem, our algorithm gives $(1+\varepsilon)A_q$ approximation.}.
It was previously
studied by Azar and Epstein~\cite{AE} and by Kumar, Marathe, Parthasarathy and Srinivasan~\cite{KMPS}. Particular, for $q\in (1,2]$ the best known approximation algorithm has performance guarantee $2^{1/q}$ \cite{KMPS} (Theorem 4.4).
We give $\sqrt[q]{(1+\varepsilon)A_q}$-approximation algorithm for any $\varepsilon>0$  substantially improving upon previous results (see Figure~\ref{fig:plot-results}).

We formulate the unrelated parallel machine scheduling problem as a boolean nonlinear program:
\begin{align}
\min \sum_{i\in [m]} \big(\sum_{j\in [n]} &p_{ij}x_{ij}\big)^q &\label{Sch_obj}\\
\sum_{i\in [m]}x_{ij}&=1, & \forall j\in [n] \label{Sch_eqn:r2}\\
x_{ij}&\in \{0,1\}, & \forall i\in [m],\, j\in [n] \label{Sch_eqn:r5}
\end{align}

Using Theorem~\ref{thm:efficient}, we obtain an almost optimal fractional solution $(x,z)$ of the LP relaxation (\ref{R:obj})--(\ref{R:rel3})
corresponding to the IP (\ref{Sch_obj})--(\ref{Sch_eqn:r5}). We use the straightforward randomized rounding:
we assign each job $j$ to machine $i$ with probability $x_{ij}$. We claim that, by Theorem~\ref{thm:main}, the expected cost of our integral solution
is upper bounded by $A_{q}$ times the value of the fractional solution $(x,z)$. Indeed, the probability that we assign a job $j$ to machine $i$ is
exactly equal to $x_{ij}$; and we assign job $j$ to machine $i$ independently of other jobs. That implies that our approximation algorithm has a performance guarantee of $\sqrt[q]{(1+\varepsilon)A_q}$ for the $\ell_q$-norm objective.

\subsection{Unrelated Parallel Machine Scheduling with Nonlinear Functions of Completion Times}
\label{section:UnrelatedParallel}
As in the previous problem, in Unrelated Parallel Machine Scheduling with Nonlinear Functions of Completion Times, we are given $n$ jobs and $m$ machines. The processing time of the job $j\in [n]$ assigned to the machine $i\in [m]$ is $p_{ij}\ge 0$. We need to assign jobs to machines and set their start times such that job processing intervals do not overlap. The goal is to minimize $\sum_{s=1}^nw_jC_j^p$ where $C_j$ is the completion time of job $j$ in the schedule and $p\ge 1$. Using classical scheduling notation this problem can be denoted as $R|| \sum_j w_j C_j^p$.

The problem $R|| \sum_j w_j C_j^p$ is well studied for $p=1$. It is known to be APX-hard \cite{HSW} while the best known approximation algorithm has a performance guarantee of $3/2$ \cite{SS,S}. For $p>1$ even the single machine scheduling problem is not understood: It is an open problem whether $1|| \sum_j w_j C_j^p$ is $\cal NP$-hard for $p>0$, $p\neq 1$. Bansal and Pruhs~\cite{BP} and Stiller and Wiese~\cite{SW}  gave constant factor approximation algorithms for more general functions of completion times for a single machine. However, there were no known approximation algorithms for multiple machines.
We show how to use our framework for this problem in Appendix~\ref{section:UnrelatedParallel-Tech}. Our algorithm gives $2^pA_p$ approximation.

\subsection{Degree Balanced Spanning Tree Problem}
\label{section:QST}

We are given an undirected graph $G=(V,E)$ with edge weights $w_e\ge 0$. The goal is to find a spanning tree $T$ minimizing the objective function
\begin{equation}\label{Tree:Objective}
f(T)=\sum_{v\in V} \left(\sum_{e\in \delta(v)\cap T} w_e \right)^q,
\end{equation}
where $\delta(v)$ is the set of edges in $E$ incident to the vertex $v$. For $q=2$, a more general problem was considered before in the Operations Research literature \cite{tree1, tree2, tree3, tree4} under the name of Adjacent Only Quadratic Spanning Tree Problem.
 A related problem, known as Degree Bounded Spanning Tree, received a lot of attention
in Theoretical Computer Science \cite{SL,G}. We are not aware of any previous work on Degree Balanced Spanning Tree Problem.

Let $x_e$ be a boolean decision variable such that $x_e=1$ if we choose edge $e\in E$ to be in our solution (tree) $T$. We formulate our problem as the following convex boolean optimization problem
\begin{eqnarray*}
\min \sum_{v\in V} \left(\sum_{e\in \delta(v)} w_ex_e\right)^q \label{Tree_obj}\\
x\in {\cal B}({\cal M}) &\label{Tree_eqn:r2}
\\x_{e}\in \{0,1\}, & \hspace{2cm} \forall e\in E, \label{Tree_eqn:r5}
\end{eqnarray*}
where ${\cal B}({\cal M})$ is the base polymatroid polytope of the graphic matroid in graph $G$. We refer the reader to Schrijver's book~\cite{S02} for the definition of the matroid. Using Theorem~\ref{thm:efficient}, we obtain an almost optimal fractional solution $x^*$ of
LP relaxation (\ref{R:obj})-(\ref{R:rel3}) corresponding to the above integer problem.

Following Calinescu et al. \cite{CCPV}, we define the continuous extension of the objective function (\ref{Tree:Objective}) for any fractional solution $x'$
$$F(x')=\sum_{S\subseteq [n]}f(S)\prod_{e\in S}x'_e\prod_{e\not\in S}(1-x'_e),$$
i.e. $F(x')$ is equal to the expected value of the objective function (\ref{Tree:Objective}) for the set of edges sampled independently at random with probabilities $x'_e,e\in E$. The function $F$ can
 be approximated with arbitrary polynomially small precision efficiently via sampling (see \cite{CCPV}).
By Theorem~\ref{thm:main}, we get the bound $F(x^*)\le A_q \cdot LP^*$,
where $LP^*$ is the value of the LP relaxation (\ref{R:obj})--(\ref{R:rel3}) on the fractional solution $x^*$.

The rounding phase of the algorithm implements the pipage rounding technique~\cite{AS} adopted to polymatroid polytopes by Calinescu et al.~\cite{CCPV}.
Calinescu et al.~\cite{CCPV} showed that given a matroid ${\cal M}$ and a fractional solution $x\in {\cal B}({\cal M})$, one can efficiently find two elements, or two edges
in our case, $e'$ and $e''$ such that the new fractional solution ${\tilde x}(\varepsilon)$ defined as ${\tilde x}_{e'}(\varepsilon)=x_{e'}+\varepsilon$, ${\tilde x}_{e''}(\varepsilon)=x_{e''}-\varepsilon$ and ${\tilde x}_{e}(\varepsilon)= x_{e}$ for $e\notin\{e',e''\}$ is feasible in the base polymatroid polytope for small positive and for small negative values of $\varepsilon$.

They also showed that if the objective function $f(S)$ is submodular then the function of one variable
$F({\tilde x}(\varepsilon))$ is convex. In our case, the objective function $f(S)$ is supermodular which follows from a more general folklore statement.
 \begin{fact}
 The function $f(S)=g(\sum_{i\in S}w_i)$ is supermodular if $w_i\ge 0$ for $i\in [n]$ and $g(x)$ is a convex function of one variable.
 \end{fact}
Therefore, the function $F({\tilde x}(\varepsilon))$ is concave. Hence, we can apply the pipage rounding directly: We start
 with the fractional solution $x^*$. At every step, we pick $e'$ and $e''$ (using the algorithm from~\cite{CCPV}) and move to ${\tilde x}(\varepsilon)$ with $\varepsilon=\varepsilon_1=-\min\{x_{e'},1-x_{e''}\} $ or $\varepsilon=\varepsilon_2=\min\{1-x_{e'},x_{e''}\} $ whichever minimizes the concave function $F({\tilde x}(\varepsilon))$ on the interval $[\varepsilon_1,\varepsilon_2]$. We stop when the current solution ${\tilde x}$ is integral.

At every step, we decrease the number of fractional variables $x_e$ by at least 1. Thus, we terminate the algorithm in at most $|E|$ iterations.
The value of the function $F({\tilde x})$ never increases. So the cost of the final integral solution is at most the cost of the
initial fractional solution $x^*$, which, in turn, is at most $A_q \cdot LP^*$.

Note, that we have not used any special properties of graphic matroids. The algorithm from~\cite{CCPV} works for general matroids accessible through oracle calls. So we can apply our technique to more general problems where the objective is to minimize a function like~(\ref{Tree:Objective})
subject to base matroid constraints.

\section{Proof of Theorem~\ref{thm:efficient}}\label{sec:thm:efficient}

We now give an efficient algorithm for finding $(1+\varepsilon)$ approximately optimal solution
to LP (\ref{R:obj})-(\ref{R:rel3}).

\begin{proof}[Proof of Theorem~\ref{thm:efficient}]

Observe that for every $y\in \calP$, there exists a $z$ such that the pair $(y,z)$ is a feasible solution to
LP (\ref{R:obj})-(\ref{R:rel3}). For example, one
such $z$ is defined as $z_{jS} = \prod_{i\in S} y_i \prod_{i\notin S} (1-y_i)$. Of course, this particular $z$ may be suboptimal. However, it turns out, as we show
below, that for every $y$, we can find the optimal $z$ efficiently. Let us denote the minimal cost of the $j$-th term in (\ref{R:obj}) for a given $y\in \calP$
by $H_j(y)$. That is, $H_j(y)$ is the cost of the following LP. The variables of the LP are $z_{jS}$. The
parameters $y\in \calP$ and $j\in [k]$ are fixed.
\begin{equation}
\min \sum_{S\subseteq [n]} f_j\Big( \sum_{i\in S} d_{ij}\Big)z_{jS} \label{R-eff:obj}
\end{equation}
\begin{eqnarray}
\sum_{S\subseteq [n]}z_{jS}=1 && \label{const1-eff}\\
\sum_{S:i\in S}z_{jS}=y_i,&& \forall i\in [n]\label{const2-eff}\\
z_{jS}\ge 0,&& \forall S\subseteq [n] \label{R-eff:rel3}
\end{eqnarray}
Now, LP (\ref{R:obj})-(\ref{R:rel3}) can be equivalently rewritten as (below $y$ is the variable).
\begin{eqnarray}
\min \sum_{j\in [k]} H_j(y) \label{LP:eff:obj}\\
y\in {\cal P}&& \label{LP:eff:poly}
\end{eqnarray}
The functions $H_j(y)$ are convex\footnote{If $z^*$ and $z^{**}$ are the optimal solutions for vectors
$y^*$ and $y^{**}$, then $\lambda z^* + (1-\lambda) z^{**}$ is a feasible solution for $\lambda y^* + (1-\lambda) y^{**}$. Hence,
$H_j(\lambda y^* + (1-\lambda) y^{**}) \leq \lambda H_j(y^*) + (1-\lambda) H_j(y^{**})$. See Section~\ref{sec:convexFj} in Appendix for details.}.
In Lemma~\ref{lem:find-Z} (see below), we prove that LP~(\ref{R-eff:obj})-(\ref{R-eff:rel3}) can be solved in polynomial time, and thus the functions
$H_j(y)$ can be computed efficiently. The algorithm for finding $H_j(y)$ also returns a subgradient of $H_j$ at $y$. Hence, the minimum of
convex problem~(\ref{LP:eff:obj})-(\ref{LP:eff:poly}) can be found using the ellipsoid method. Once the optimal $y^*$ is found, we
find $z^*$ by solving LP~(\ref{R-eff:obj})-(\ref{R-eff:rel3}) for $y^*$ and each $j\in [k]$.
\end{proof}

\begin{lemma}\label{lem:find-Z}
There exists a polynomial time algorithm for computing $H_j$ and finding a subgradient of $H_j$.
\end{lemma}
\begin{proof}
We need to solve LP~(\ref{R-eff:obj})-(\ref{R-eff:rel3}). Recall that  in Section~\ref{sec:discr} (Theorem~\ref{thm:discr}) we show how to choose $\delta_j>0$
 such that each $d_{ij}/\delta_j$ is an integer polynomially bounded.  For simplicity we assume that $\delta_j=1$ (the proof in general case is almost identical). Therefore  $d_{ij}$ are integral in this case and polynomially bounded.
We write the dual LP. We introduce a variable $\xi$ for constraint~(\ref{const1-eff})
and variables $\eta_i$ for constraints~(\ref{const2-eff}).
\begin{eqnarray}
\max \;\;\xi + \sum_{i} \eta_i y_i && \label{dual:obj}\\
\xi + \sum_{i\in S} \eta_i \leq f_j \Big(\sum_{i\in S} d_{ij}\Big), && \forall S\subset [n]\label{dual:constr}
\end{eqnarray}
The LP has exponentially many constraints. However, finding a violated constraint is easy. To do so, we guess $B^*= \sum_{i\in S^*} d_{ij}$
for the set $S^*$ violating the constraint. That is possible, since all $d_{ij}$ are polynomially bounded, and so is $B^*$.
Then we solve the maximum knapsack problem
\begin{eqnarray*}
\max_{S\subseteq [n]} \sum_{i\in S}\eta_{i}\\
\sum_{i\in S} d_{ij}=B^*
\end{eqnarray*}
using the standard dynamic programming algorithm and obtain the optimal set $S^*$. The knapsack problem is polynomially solvable, since $B^*$ is polynomially bounded. If
$\xi + \sum_{i\in S^*} \eta_{i} > f_j\left(B^*\right)$, then constraint (\ref{dual:constr}) is violated for the set $S^*$; otherwise all
constraints (\ref{dual:constr}) are satisfied.

Let $(\xi^*,\eta^*)$ be the optimal solution of the dual LP. The value of the function $H_j(y)$ equals the objective value of the dual LP.
A subgradient of $H_j$ at $y$ is given by the equation
\begin{equation}\label{eq:subgrad}
\tilde{y}\mapsto \xi^* + \sum_i \eta^*_i \tilde{y}_i.
\end{equation}
This is a subgradient of $H_j$, since $(\xi^*,\eta^*)$ is a feasible solution of the dual LP for every $\tilde{y}$ (note that constraint (\ref{dual:constr}) does not depend on $y$), and, hence,
(\ref{eq:subgrad}) is a lower bound on $H_j(\tilde{y})$.
\end{proof}

\section{Proof of Theorem~\ref{thm:main}}\label{sec:thm:main}
In this section, we prove the main theorem -- Theorem~\ref{thm:main}.
\begin{proof}[Proof of Theorem~\ref{thm:main}]
The theorem easily follows from the de la Pe\~na decoupling inequality (Theorem~\ref{thm:decoupling} and Corollary~\ref{cor:neg-assoc}) for $f_j(t) = c_j t^q$, and 
from the more general inequality presented in Theorem~\ref{thm:CX-main} (see also Corollary~\ref{cor:A-f-inequality}) for arbitrary convex functions $f_j$.
Consider a feasible solution $(y^*,z^*)$ to
IP (\ref{IP:obj})-(\ref{IP:int}). We prove inequality~(\ref{eq:thm:main}) term by term. That is,
for every $j$ we show that
\begin{equation}\label{eq:thm:main:need-to-prove}
\Exp\Big[f_j\Big(\sum_{i\in D_j} d_{ij} R_i(y^*)\Big) \Big]\leq
A(f_j) \sum_{S\subseteq [n]} f_j\Big( \sum_{i\in S\cap D_j} d_{ij}\Big)z^*_{jS}.
\end{equation}
Recall that $D_j = \{i: d_{ij}\neq 0\}$. Above, we dropped terms with $i\notin D_j$, since if $i\notin D_j$, then $d_{ij}=0$.

Fix a $j\in [n]$. Define random variables $Y_i$ for $i\in D_j$ as follows: Pick a random set $S\subset [n]$ with probability $z_{jS}$, and let
$Y_i = d_{ij}$ if $i\in S$, and $Y_i=0$ otherwise. Note that random variables $Y_i$ are dependent. We have
$$\Pr(Y_i = d_{ij}) = \sum_{S: i\in S}\Pr(S) = \sum_{S: i\in S} z^*_{jS} = y^*_i.$$
It is easy to see that
$$\Exp\Big[f_j\Big(\sum_{i\in [n]} Y_i \Big)\Big] = \sum_{S\subseteq [n]} f_j\Big( \sum_{i\in S} d_{ij}\Big) z_{jS}.$$
The right hand side is simply the definition of the expectation on the left hand side. Now, let $X_i = d_{ij}\,R_i(y^*)$ for $i\in D_j$. Note that
by conditions of the theorem, $\Pr(X_i = d_{ij}) = y^*_i = \Pr (Y_i = d_{ij})$ (by condition (1)). Thus, each $X_i$ has the same distribution as $Y_i$.
Furthermore, $X_i$'s are independent or negatively associated (by condition (2)). Therefore, we can apply the decoupling inequality from Theorem~\ref{thm:CX-main} (see Corollary~\ref{cor:A-f-inequality})
$$\Exp\Big[f_j\Big(\sum_{i\in D_j} X_i\Big)\Big] \leq A(f_j) \Exp\Big[f_j\Big(\sum_{i \in D_j} Y_i\Big)\Big].$$
The left hand side of the inequality equals the left hand side of (\ref{eq:thm:main:need-to-prove}),
the right hand side of the inequality equals the right hand side of (\ref{eq:thm:main:need-to-prove}). Hence, inequality (\ref{eq:thm:main:need-to-prove})
holds.
\end{proof}

\section{Decoupling Inequality}\label{sec:thm:decoupling}
In this section, we prove the decoupling inequality (Theorem~\ref{thm:decoupling}) with the optimal constant $C_q=A_q^{\nicefrac{1}{q}}$. In fact, we prove a more general inequality which works for arbitrary convex functions. To state the inequality we need the notion of
\emph{convex stochastic order}.

\begin{definition}
We say that a random variable $X$ is less than $Y$ in the convex (stochastic) order, and write
$X\cx Y$ if for every convex function $\varphi:\bbR\to \bbR$,
\begin{equation}\label{eq:def-cx}
\Exp \varphi(X) \leq \Exp\varphi(Y),
\end{equation}
whenever both expectations exist.
\end{definition}
\begin{remark}
If $X\cx Y$ and $\psi$ is a \emph{concave} function, then $\Exp[\psi(X)]\geq \Exp[\psi(Y)]$, since the function $\phi(x)=-\psi(x)$ is convex, and therefore
$-\Exp[\psi(X)]\leq -\Exp[\psi(Y)]$.
\end{remark}
It is easy to see that the convex stochastic order defines a partial order on all random variables. Particularly,
if $X\cx Y$ and $Y\cx Z$, then $X\cx Z$. Note that the definition depends only on the distributions of $X$ and $Y$.
That is, if $X$ has the same distribution as $Z$ and $X\cx Y$, then $Z\cx Y$. The random variables $X$ and $Y$ may
be defined on the same probability space or on different probability spaces. We refer the reader to the book of
Shaked and Shanthikumar~\cite{SS-Book} for a detailed introduction to stochastic orders.

We now state the general inequality in terms of the convex order. Part II of the theorem shows that the inequality is
tight.

\begin{theorem}\label{thm:CX-main}
I. Let $Y_1,\dots, Y_n$ be jointly distributed nonnegative (non-independent) random variables, and
let $X_1,\dots, X_n$ be independent random variables such that
each $X_i$ has the same distribution as $Y_i$. Let $P$ be a Poisson random variable with parameter 1 independent
of $Y_i$'s. Then,
\begin{equation}\label{eq:CX-main}
\sum_{i=1}^n X_i\cx P \sum_{i=1}^n Y_i.
\end{equation}

II. For every nonnegative function $\varphi$ with a finite
expectation $\Exp[\varphi(P)]$ and every positive $\varepsilon$, there exists $n$ and
random variables $X_1,\dots,X_n$, $Y_1,\dots, Y_n$ as in part I such that
\begin{equation}\label{eq:CX-counterexample}
\E{\varphi\Big(\sum_{i=1}^n X_i \Big)}\geq (1+\varepsilon) \E{\varphi\Big(P \sum_{i=1}^n Y_i \Big)}.
\end{equation}
\end{theorem}

\begin{proof}[Proof of Theorem~\ref{thm:decoupling}]
Theorem~\ref{thm:CX-main} implies Theorem~\ref{thm:decoupling}, since the function $t\to t^q$ is convex and thus
\begin{align*}
\|X_1+\dots+X_n\|_q^q &= \E{ (X_1+\dots+X_n)^q} \leq \E{\big(P\cdot (Y_1+\dots+Y_n)\big)^q}\\
&= \E{P^q}\cdot \E{(Y_1+\dots+Y_n)^q} = A_q \E{(Y_1+\dots+Y_n)^q}\\& = A_q\|Y_1+\dots+Y_n\|_q^q.
\end{align*}
Part II of Theorem~\ref{thm:CX-main} shows that we cannot replace $A_q$ with a smaller constant.
\end{proof}
Another immediate corollary of Theorem~\ref{thm:CX-main} is as follows.
\begin{corollary}\label{cor:A-f-inequality}
For an arbitrary nonnegative convex function $f$,
$$\Exp[f(X_1+\dots+X_n)] \leq A(f) \Exp[f(Y_1+\dots+Y_n)],$$
where $A(f)$ is defined in (\ref{eq:def-general-A}).
\end{corollary}
\begin{proof}
Write,
$$\Exp[f(X_1+\dots+X_n)] \leq \Exp\big[f\big(P \cdot(Y_1+\dots+Y_n)\big)\big] =
\Exp\Exp\big[f\big(P \cdot(Y_1+\dots+Y_n)\big)\given Y_1+\dots+Y_n\big].$$
For any $t>0$, particularly for $t=\sum_i Y_i$, we have
$\Exp[f(tP)]\leq A(f) f(t)$, hence
$$\Exp[f(X_1+\dots+X_n)] \leq \Exp[A(f) f(Y_1+\dots+Y_n)] = A(f) \Exp[f(Y_1+\dots+Y_n)].$$
\end{proof}

We first prove part II of Theorem~\ref{thm:CX-main}. Consider the following example. Let $Y^{(n)}_i, i\in \{1,\dots,n\}$, be random variables taking
value $1$ with probability $1/n$, and $0$ with probability $1-1/n$. We generate $Y_i^{(n)}$'s as follows.
We pick a random $j\in[n]$ and
let $Y_j^{(n)}=1$ and $Y_i^{(n)}=0$ for $i\neq j$. Random variables $X^{(n)}_i$ are i.i.d. Bernoulli random variables with $\Exp[X^{(n)}_i]=1/n$.
Then, the sum $\sum_{i=1}^n Y^{(n)}_i$ always equals 1, and $\E{\varphi(P \sum_{i=1}^n Y^{(n)}_i)} = \E{\varphi(P)}$.
As $n\to \infty$, the sum $\sum_{i=1}^n X^{(n)}_i$ converges in distribution to $P$ (by the Poisson limit theorem).
Thus (see Lemma~\ref{lem:limsup} for details),
$$\sup_n\E{\varphi\Big(\sum_{i=1}^n X^{(n)}_i\Big)} \geq \E{\varphi(P)},$$
and, hence, for some $n$ inequality~(\ref{eq:CX-counterexample}) holds.

Before proceeding to the proof of part I, we state some known properties of the convex order.

\begin{lemma}[Theorem 3.A.12 in~\cite{SS-Book}]\label{lem:sumX-less-sumY}
Suppose $X_1,\dots, X_n$ are
independent random variables and $Y_1,\dots, Y_n$ are
independent random variables. If $X_i \cx Y_i$ for all $i$, then
$\sum_{i=1}^n X_i\cx \sum_{i=1}^n Y_i$.
\end{lemma}

\begin{lemma}[Theorem 3.A.36 in~\cite{SS-Book}]\label{lem:sum-aX-less-aY}
Consider random variables $X_1,\dots,X_n$ and $Y$. If $X_i\cx Y$ for all $i$, then
$\sum_{i=1}^n a_i X_i \leq \sum_{i=1}^n a_i Y$,
for any sequence of nonnegative numbers $a_1,\dots, a_n$.
\end{lemma}

For completeness, we prove Lemma~\ref{lem:sumX-less-sumY} and Lemma~\ref{lem:sum-aX-less-aY} in Appendix. To simplify the proof,
we will use the following easy lemma.
\begin{lemma}\label{lem:simplify}
If for random variables $X$ and $Y$ condition (\ref{eq:def-cx}) is satisfied for all convex functions $\varphi$ with $\varphi(0)=0$, then
$X\cx Y$.
\end{lemma}
\begin{proof}
Consider an arbitrary convex function $\varphi$. Let $\tilde{\varphi}(x) = \varphi(x) -\varphi(0)$. Since $\tilde{\varphi}(0)=0$, we have
$$\Exp[\varphi(X)] =  \Exp[\tilde{\varphi}(X)] + \varphi(0) \leq  \Exp[\tilde{\varphi}(Y)] + \varphi(0) = \Exp[\varphi(Y)].$$
\end{proof}

First, we prove that a Bernoulli random variable $B$ can be upper bounded in the convex order by a
Poisson random variable $P$ with $\Exp[P]  = \Exp[B]$.
\begin{lemma}\label{lem:B-cx-P}
Let $P$ be an integral random variable, and $B$ be a Bernoulli random variable with parameter $p = \Exp[P]$. Then,
$B\leq_{cx}P$. Particularly, if $P$ is a Poisson random variable with $\Exp[P]=\Exp[B]$, then $B\leq_{cx}P$.
\end{lemma}
\begin{proof}
Consider an arbitrary convex function $\varphi:\bbR\to\bbR$ with $\varphi(0)=0$ (see Lemma~\ref{lem:simplify}). Define linear function $l:\bbR\to\bbR$ as
$l(x)=\varphi(1)x$. The graph of $l$ intersects the graph of $\varphi$ at points $(0,0)$ and $(1, \varphi(1))$. Since $\varphi$ is convex,
we have $l(x)\leq \varphi(x)$ for $x\notin (0,1)$. Hence, for every integral $k$, $l(k)\leq \varphi(k)$, and $l(P)\leq \varphi(P)$. Consequently,
$$\E{\varphi(P)} \geq \E{l(P)} = l(\E{P}) = l(\E{B}) = \varphi(1) \E{B} = \Exp[\varphi(B)].$$
\end{proof}

Now we consider a very special case of Theorem~\ref{thm:CX-main} when all $X_i$'s and $Y_i$'s are
Bernoulli random variables scaled by a factor $\alpha_i$, and all events $\{Y_i=1\}$ are mutually exclusive.
As we see later, the general case can be easily reduced to this special case.

\begin{lemma}\label{lem:main-CX-lemma}
Consider Bernoulli random variables $\chi_1,\dots, \chi_n$ such that
$\sum_{i=1}^n \Pr(\chi_i = 1) = 1$,
and all events $\{\chi_i = 1\}$ are mutually exclusive. (That is, with probability 1 one and only one $\chi_i$ equals~1.)
Let $B_1,\dots, B_n$ be independent Bernoulli random variables such that $\Pr(B_i = 1) = \Pr(\chi_i = 1)$.
Then, for all nonnegative numbers $\alpha_1,\dots, \alpha_n$, we have
$$\sum_{i=1}^n \alpha_i \chi_i \cx \sum_{i=1}^n \alpha_i B_i \cx P \sum_{i=1}^n \alpha_i \chi_i,$$
where $P$ is a Poisson random variable with parameter 1 independent of $\chi_i$'s.
\end{lemma}
\begin{proof}
I. We prove the first inequality. Consider an arbitrary convex function $\varphi:\bbR\to \bbR$ with
$\varphi(0) = 0$. The function $\varphi$ is superadditive i.e.,
$\varphi (a + b) \geq \varphi (a) + \varphi(b)$ for all positive $a$ and $b$ (since the derivative of $\varphi$ is monotonically
non-decreasing, and $\varphi(0)=0$). Hence, we have
$$\varphi\big(\sum_{i=1}^n \alpha_i B_i\big)\geq \sum_{i=1}^n \varphi(\alpha_i B_i) = \sum_{i=1}^n B_i \varphi(\alpha_i),$$
and
\begin{multline*}
\E{\varphi\big(\sum_{i=1}^n \alpha_i B_i\big)}\geq \E{\sum_{i=1}^n B_i \varphi(\alpha_i)} =
\sum_{i=1}^n\Pr(B_i = 1) \varphi(\alpha_i) \\= \sum_{i=1}^n\Pr(\chi_i = 1) \varphi(\alpha_i) = \E{\varphi\big(\sum_{i=1}^n \alpha_i \chi_i\big)}.
\end{multline*}

II. We prove the second inequality. Using Lemma~\ref{lem:B-cx-P} and Lemma~\ref{lem:sumX-less-sumY}, we replace
Bernoulli random variable $B_1,\dots, B_n$ with independent Poisson random variables $P_1,\dots, P_n$
satisfying $\Exp\, P_i = \Exp\, B_i$. The sum $P_1+\dots + P_n$ is distributed as a Poisson
random variable with parameter~1. By coupling random variables
$(P_1+\dots + P_n)$ and $P$, we may assume that $P=P_1+\dots + P_n$.

Consider an arbitrary convex function $\varphi$ with $\varphi(0) = 0$. Using convexity of the function $\varphi$, we derive
\begin{align*}
\E{\varphi\big(\sum_{i=1}^n \alpha_i P_i \big)}
&=
\sum_{k=1}^{\infty}
\E{\varphi\big(\sum_{i=1}^n \alpha_i P_i \big) \given P = k}\Pr(P  = k)\\
&=
\sum_{k=1}^{\infty}
\E{\varphi\big(\sum_{i=1}^n \frac{P_i}{k} \cdot \alpha_i k \big) \given P = k}\Pr(P  = k)\\
&\leq
\sum_{k=1}^{\infty}
\E{ \sum_{i=1}^n \frac{P_i}{k} \varphi(\alpha_i k)  \given P = k}\Pr(P  = k)\\
&=
\sum_{k=1}^{\infty} \sum_{i=1}^n \E{\frac{P_i}{k}\given P = k}
\varphi(\alpha_i k)\Pr(P=k).
\end{align*}

We observe that $\E{P_i\given P = k}= k\,\E{P_i}$, which follows from the following well known fact (see e.g., Feller~\cite{Feller}, Section IX.9, Problem~6(b), p.~237).

\begin{fact}\label{fact:p12} Suppose $P_a$ and $P_b$ are independent Poisson random variables with parameters $\lambda_a$ and $\lambda_b$. Then, for every $k\in \bbN$,
$\Exp\big[P_a \given P_a + P_b = k\big] = \frac{\lambda_a}{\lambda_a+\lambda_b}\;k$.
\end{fact}
In our case, $P_a = P_i$, $P_b=\sum_{i'\neq i} P_{i'}$, $P_a + P_b = P$.
Therefore, we have
\begin{align*}
\Exp\Big[\varphi\big(\sum_{i=1}^n \alpha_i P_i \big)\Big]
&\leq
\sum_{k=1}^{\infty} \sum_{i=1}^n \Exp\Big[\frac{P_i}{k}\given P = k\Big]
\varphi(\alpha_i k)\Pr(P=k)\\
&=
\sum_{k=1}^{\infty} \sum_{i=1}^n \E{P_i}
\varphi(\alpha_i k)\Pr(P=k)\\
&=
\sum_{i=1}^n \E{\chi_i} \sum_{k=1}^{\infty}
\varphi(\alpha_i k)\Pr(P=k)\\
&=
\sum_{i=1}^n \E{\chi_i} \E{\varphi(\alpha_i k)}\\
&=
\Exp\Big[\varphi(P \sum_{i=1}^n \alpha_i \chi_i)\Big].
\end{align*}
\end{proof}

\begin{proof}[Proof of Theorem~\ref{thm:CX-main}]
Denote by $\calY$ the support of the random vector $Y=(Y_1,\dots, Y_n)$. Each $Y_i$ can be represented as follows
$$Y_i = \sum_{y\in \calY} y_i\, \chi(Y=y),$$
where $\chi(Y=y)$ is the indicator of the event $\{Y=y\}$. Here we assume that $\calY$ is finite.
We treat the general case in Appendix~\ref{sec:crv}.
Applying Lemma~\ref{lem:main-CX-lemma} to the random variables $\chi(Y=y)$, we get
$$Y_i = \sum_{y\in \calY} y_i\, \chi(Y=y)\cx \sum_{y\in \calY} y_i\, B_y^i,$$
where $B_y^i$ are independent Bernoulli random variables. Each $B_y^i$ is distributed as the random variable
$\chi(Y=y)$; that is, $\Pr (B_y^i = 1)= \Pr(Y=y)$. Since each $X_i$ has the same distribution as $Y_i$, we have
$$X_i \cx \sum_{y\in \calY} y_i\, B_y^i.$$
By Lemma~\ref{lem:sumX-less-sumY}, we can sum up this inequality over all $i$ from $1$ to $n$:
$$\sum_{i=1}^n X_i \cx \sum_{i=1}^n \sum_{y\in \calY} y_i\, B_y^i = \sum_{y\in \calY} \Big(\sum_{i=1}^n y_i\, B_y^i\Big).$$
We apply Lemma~\ref{lem:sum-aX-less-aY} to every sum in parentheses:
$$\sum_{i=1}^n y_i\, B_y^i\cx \sum_{i=1}^n y_i\, B_y.$$
Again, using Lemma~\ref{lem:sumX-less-sumY}, we get
$$\sum_{i=1}^n X_i \cx\sum_{y\in \calY} \Big(\sum_{i=1}^n y_i\, B_y^i\Big)\cx
\sum_{y\in \calY}\Big(\sum_{i=1}^n y_i\Big)\, B_y.$$
Finally, by Lemma~\ref{lem:main-CX-lemma},
$$\sum_{y\in \calY}\Big(\sum_{i=1}^n y_i\Big)\, B_y\cx
P\cdot \sum_{y\in \calY}\Big(\sum_{i=1}^n y_i\Big)\, \chi (Y=y) =
P\cdot \sum_{i=1}^n \Big(\sum_{y\in \calY} y_i\, \chi (Y=y) \Big)
=P\sum_{i=1}^n Y_i.$$
This concludes the proof of Theorem~\ref{thm:CX-main}.
\end{proof}

\section{Negatively Associated Random Variables}\label{sec:neg-depend}
The decoupling inequalities~(\ref{eq:de-la-Pena}) and~(\ref{eq:CX-main})
can be extended to  \emph{negatively associated} random variables $X_1,\dots,X_n$.
The notion of negative association is defined as follows.
\begin{definition}[Joag-Dev and Proschan~\cite{JP}]
Random variables $X_1,\dots, X_n$ are negatively associated if for all disjoint sets $I, J \subset [n]$ and all non-decreasing functions
$f:\bbR^I\to \bbR$ and $g:\bbR^J\to \bbR$ the following inequality holds:
$$
\E{f(X_i, i\in I) \cdot g(X_j, j\in J)} \leq \E{f(X_i, i\in I)} \cdot \E{g(X_j, j\in J)}.
$$
\end{definition}

Shao~\cite{Shao} showed that if $X_1,\dots,X_n$ are negatively associated random variables, and $X^*_1,\dots,X^*_n$ are independent random variables such that
each $X^*_i$ is distributed as $X_i$, then for every convex function $\varphi:\bbR\to \bbR$,
$$\E{ \varphi (X_1+\dots+X_n)} \leq \E{ \varphi (X^*_1+\dots+X^*_n)}.$$
In other words, $X_1+\dots+X_n\cx X^*_1+\dots+X^*_n$ (see also Theorem 3.A.39 in \cite{SS-Book}).

\begin{corollary}\label{cor:neg-assoc}
Let $Y_1,\dots, Y_n$ be jointly distributed nonnegative (non-independent) random variables, and
let $X_1,\dots, X_n$ be negatively associated random variables such that
each $X_i$ has the same distribution as $Y_i$. Let $P$ be a Poisson random variable with parameter 1 independent of the random variables $Y_i$. Then,
$$\sum_{i=1}^n X_i\cx P \sum_{i=1}^n Y_i.$$
Particularly, for every convex nonnegative $f$,
$$f\big(\sum_{i=1}^n X_i\big) \leq A(f) f\big(\sum_{i=1}^n Y_i\big),$$
and for $q \geq 1$,
$$\big \|\sum_{i=1}^n X_i\big\|_{q} \leq A_{q}^{\nicefrac{1}{q}} \,\big \|\sum_{i=1}^n Y_i\big\|_{q},$$
where $A(f)$ is defined in (\ref{eq:def-general-A}), and $A_{q}$ is the fractional Bell number.
\end{corollary}

\section{Generalizations}\label{sec:Generalizations}
We can extend our results to maximization problems with the objective function
\begin{equation}
\sum_{j\in [k]} f_j\big(\sum_{i\in [n]} d_{ij} y_{i}\big),\label{eq:gener}
\end{equation}
if $f_j$'s are arbitrary non-decreasing nonnegative \textit{concave} functions defined on $\bbR^{\geq 0}$. The approximation
ratio equals $\min_j B(f_j)$, where
$$B(f_j) = \inf_{t>0} \Exp\Big[\frac{f_j(Pt)}{f_j(t)}\Big].$$
It is not hard to see that $B(f)\geq 1-1/e$ for all $f$. Indeed, if $P\geq 1$, then
$f(Pt)/f(t) \geq 1$, thus $B(f)= \Exp[f(Pt)/f(t)] \geq \Pr(P\geq 1) = 1 - 1/e$. This bound is 
tight if $f(Pt)/f(t) = 1$ for $P\geq 1$. For example, $B(f) = 1-1/e$ for the function $f(t) = \min\{t,1\}$. Note that the approximation ratio of $1-1/e \approx 0.632$ for maximization problems of this form
was previously known (see Calinescu et al.~\cite{CCPV}).
However, for some concave functions $f$ we get a better approximation. For example, for
$f(t) = \sqrt{t}$, we get an approximation ratio of $B(\sqrt{t}) \approx 0.773$.

\section*{Acknowledgment}
We would like to thank the anonymous referees for valuable comments.

\pagebreak
\appendix

\section{Corollary~\ref{cor:decoupling}}
In this section, we prove a corollary of Theorem~\ref{thm:CX-main}, which we will need in the next section.
\begin{corollary}\label{cor:decoupling}
Let $Y_1,\dots, Y_n$ be jointly distributed (non-independent) nonnegative integral random variables, and
let $X_1,\dots, X_n$ be independent Bernoulli random variables taking values $0$ and $1$ such that
for each $i$, $\Exp [X_i] = \Exp[Y_i]$. Then,
$$X_1+\dots+ X_n\cx P(Y_1 +\dots + Y_n).$$
Particularly, for every $q \geq 1$,
$$\big \|\sum_{i=1}^n X_i\big\|_{q} \leq A_{q}^{\nicefrac{1}{q}} \,\big \|\sum_{i=1}^n Y_i\big\|_{q},$$
where $A_{q}$ is the fractional Bell number.
\end{corollary}
\begin{proof}
Consider independent random variables $X^*_i$ such that each $X^*_i$ is distributed as $Y_i$. By
Lemma~\ref{lem:B-cx-P}, $X_i \cx X^*_i$ for all $i$. Hence, by Lemma~\ref{lem:sumX-less-sumY} and  Theorem~\ref{thm:CX-main},
$$X_1+\cdots+X_n\cx X^*_1+\cdots+X^*_n \cx Y_1+\cdots+Y_n.$$
\end{proof}

\section{Unrelated Parallel Machine Scheduling with Nonlinear Functions of Completion Times -- Technical Details}
\label{section:UnrelatedParallel-Tech}
We consider the following linear programming relaxation of the scheduling problem $R|| \sum_j w_j C_j^p$. The variable $x_{ijt} =1$ if job $j$ starts at time $t$ on machine $i$. For convenience we assume that $x_{ijt}=0$ for negative values of $t$.
 \begin{eqnarray}
\min \sum_{i\in [m]} \sum_{j\in [n]}\sum_{t\ge 0} w_j(t+p_{ij})^p x_{ijt} \label{obj_c}\\\
\sum_{i\in [m], t\ge 0}x_{ijt}=1 & \hspace{2cm} \forall j\in [n] \label{job_c}\\
\sum_{j\in [n]}\sum_{\tau=t-p_{ij}+1}^tx_{ij\tau}\le 1 & \hspace{2cm} \forall i\in [m], t\ge 0 \label{time_c}\\
x_{ijt}\ge 0& \hspace{2cm} \forall i\in [m], j\in [n], t\ge 0. \label{fr_c}
\end{eqnarray}

The constraints (\ref{job_c}) say that each job must be assigned, the constraint (\ref{time_c}) says that at most one job can be processed in a unit time interval on each machine. Such linear programming relaxation are known under the name of {\it strong time indexed formulations}. The standard issue with such relaxations is that they have pseudo-polynomially many variables due to potentially large number of indices $t$. One way to handle this issue is to partition the time interval into intervals $((1+\varepsilon)^k,(1+\varepsilon)^{k+1}]$ and round all completion times to the endpoints of such intervals. This method leads to polynomially sized linear programming relaxations with $(1+O(\varepsilon))$-loss in the performance guarantee (see \cite{S1} for detailed description of the method). From now on we ignore this issue and assume that the planning horizon upper bound $\sum_{i,j}p_{ij}$ is polynomially bounded in the input size.

\textbf{Algorithm.} Our approximation algorithm solves linear programming relaxation (\ref{obj_c})-(\ref{fr_c}). Let $x^*$ be the optimal fractional solution of the LP. Each job is tentatively assigned to machine $i$ to start at time $t$ with probability $x^*_{ijt}$, independently at random. Let $t_j$ be the tentative start time assigned to job $j$ by our randomized procedure. We process jobs assigned to each machine in the order of the tentative completion times $t_j+p_{ij}$.

\textbf{Analysis.} We estimate the expected cost of the approximate solution returned by the algorithm. We denote the expected cost by $APX$.
For each machine-job-tentative time triple $(i,j,t)$, let $J_{ijt}$ be the set of triples $(i, j',t')$ such that $t'+p_{ij'}\le t+p_{ij}$.
Let $X^{ijt}$ be the random boolean variable such that $X^{ijt}=1$ if job $j$ is assigned to machine $i$ with tentative start time $t$.
In addition, let $Z^{ijt}_{j'}$ be the random boolean variable such that $Z^{ijt}_{j'}=1$ if job $j'$ is assigned to machine $i$ with tentative start time $t'$ for some $(i,j',t')\in J_{ijt}$ by our randomized rounding procedure. Then,
$$\Pr\Big(Z^{ijt}_{j'}=1\Big)=\sum_{t':(i,j',t')\in J_{ijt}} x^*_{ij't'}.$$
Suppose that job $j$ is tentative scheduled on machine $i$ at time $t$ i.e., $X^{ijt}=1$. We start processing job $j$ after all jobs $j'$
tentative scheduled on machine $i$ at time $t'$ with $t'+p_{ij'}\leq t + p_{ij}$ are finished. Thus the weighted expected completion time to the power of $p$
for $j$ equals (given $X^{ijt}=1$)
\begin{eqnarray*}
\E{w_j C_j^p\Given X^{ijt} = 1} &=& \Exp\Big[ \Big(\sum_{j'\in [n]\setminus \{j\}} p_{i,j'}Z^{ijt}_{j'}+p_{ij} \Big)^p \Given X^{ijt}=1\Big]\\
&=& \Exp\Big[ \Big(\sum_{j'\in [n]\setminus \{j\}} p_{i,j'}Z^{ijt}_{j'}+p_{ij} \Big)^p\Big].
\end{eqnarray*}
In the second equality, we used that random variables $Z^{ijt}_{j'}$ are independent from the random variable $X^{ijt}$.
Then,
\begin{eqnarray}
APX&=& \sum_{j\in [n]} \sum_{i\in [m]} \sum_{t\ge 0} w_j\E{w_j C_j^p\Given X^{ijt} = 1}
\Pr\big( X^{ijt}=1\big) \nonumber
\\
&=&\sum_{j\in [n]} \sum_{i\in [m]} \sum_{t\ge 0} w_j\Exp \Big[\Big(\sum_{ j'\in [n]\setminus \{j\}} p_{i,j'}Z^{ijt}_{ij'}+p_{ij} \Big)^p \Big] \;x^*_{ijt}. \label{UnrelatedFirst}
\end{eqnarray}
Note, that for fixed $ i\in [m], j\in [n],t\ge 0$ random variables $Z^{ijt}_{j'}$ are independent from each other.
We claim that
\begin{equation}\label{MainUnrelated}
\Exp\Big[\Big(\sum_{ j'\in [n]\setminus \{j\}} p_{i,j'}Z^{ijt}_{j'}+p_{ij} \Big)^p \Big] \le A_p (t+2p_{ij})^p.
\end{equation}
Combining (\ref{UnrelatedFirst}) and (\ref{MainUnrelated}), we derive that the performance guarantee of our approximation algorithm is at most $2^pA_p$. We now prove inequality (\ref{MainUnrelated}).

Let $G_{ijt}$ be the interval graph where the vertex set $V(G_{ijt})$ is the collection of intervals corresponding to triples in $(i,j',t')\in J_{ijt}$ such that $x^*_{ij't'}>0$. More precisely, every triple $(i,j',t')\in J_{ijt}$ corresponds to the interval $I_{ij't'}=[t',t'+p_{ij'})$ with corresponding weight $x^*_{ij't'}>0$. Let ${\cal I}$ be the collection of all independent sets in $G_{ijt}$.
The interval graph $G_{ijt}$ is perfect, and the weights $x^*_{ij't'}$
satisfy the constraints (\ref{time_c}), so
there is a collection of weights $\lambda_C\geq 0$, $C\in {\cal I}$ (for more formal argument see below) such that
 \begin{eqnarray*}
 \sum_{C\in {\cal I}}\lambda_C=1,&&\\
\sum_{C\in {\cal I}: I_{ij't'} \in C}\lambda_C=x^*_{ij't'}, && \forall (i,j',t')\in J_{ijt}
 \end{eqnarray*}

Formally, the claim above follows from the polyhedral characterization of perfect graphs proved by Fulkerson \cite{fulkerson} and Chvatal \cite{chvatal}
(see also Schrijver's book \cite{S98}, Section 9, Application 9.2 on p. 118) that
a graph $G$ is perfect if and only if its stable set polytope is defined by the system below:
\begin{eqnarray*}
\sum_{v\in C}x_v \le 1,&& \mbox{ for each clique }C,\\
 x(v)\ge 0,&& \mbox{ for each } v\in V.
\end{eqnarray*}
In the interval graph $G_{ijt}$ all clique inequalities are included in the constraints (\ref{time_c}) and therefore any set of weights $x^*_{ij't'}$ can be decomposed into a convex combination of independent sets in $G_{ijt}$.

We define a random variable $Y^{ijt}_{j'}$ as follows: Sample an independent set $C\in {\cal I}$ with probability $\lambda_C$ and let
$$Y^{ijt}_{j'}=\left|\{I_{ij't'} \in C \}\right|.$$
Note that one job $j'$ may have more than one interval $I_{ij't'}$ in the set $C$ (for different $t'$). Random variables $Y^{ijt}_{j'}$ may
be dependent but
$$\Exp [Y^{ijt}_{j'}]=\sum_{(i,j',t')\in J_{ijt}} x^*_{ij't'} = \Exp[Z_{j'}^{ijt}].$$
Therefore, by Corollary~\ref{cor:decoupling} we have
\begin{equation}
\Exp\Big[\Big(\sum_{ j'\in [n]\setminus \{j\}} p_{i,j'}Z^{ijt}_{j'}+p_{ij} \Big)^p \Big] \le A_p\Exp\Big[\Big(\sum_{ j'\in [n]\setminus \{j\}} p_{i,j'}Y^{ijt}_{j'}+p_{ij} \Big)^p \Big].
\end{equation}
Now, observe, that $\sum_{ j'\in [n]\setminus \{j\}} p_{i,j'}Y^{ijt}_{j'}$ is always bounded by $t+p_{ij}$, because all intervals
in $C$ are disjoint ($C$ is an independent set) and all intervals are subsets of $[0,t+p_{ij}]$. Hence,
$$\Exp\Big[\Big(\sum_{ j'\in [n]\setminus \{j\}} p_{i,j'}Z^{ijt}_{j'}+p_{ij} \Big)^p \Big] \le A_p \Exp[((t + p_{ij}) + p_{ij})^p] \leq
A_p (t+2p_{ij})^p,$$
which concludes the proof.

\section{Convexity of $H_j$}\label{sec:convexFj}
We show that functions $H_j$ defined in Section~\ref{sec:thm:efficient} are convex.

\begin{lemma}\label{lem:F-convex}
Fix real numbers $q\geq 1$ and $d_1,\dots, d_n\geq 0$. Define a function $H:[0,1]^n \to \bbR^+$ as follows:
$H(y)$ equals the optimal value of the following LP:
\begin{align*}
\min \sum_{S\subseteq [n]} \big( \sum_{i\in S} d_i\big)^{q}&z_{S}&\\
\sum_{S\subseteq [n]}z_{S}&=1&\\
\sum_{S:i\in S}z_{S}&=y_i,&\forall i\in [n]\\
z_{S}&\ge 0,& \forall S\subseteq [n]&
\end{align*}
Then, $H$ is a convex function.
\end{lemma}
\begin{proof}
Consider two vectors $y^*, y^{**}\in [0,1]^n$. Pick an arbitrary $\lambda\in [0,1]$. We need to show that
$$H(\lambda y^* + (1-\lambda) y^{**}) \leq \lambda H(y^*) + (1-\lambda) H(y^{**}).$$
Consider the optimal LP solutions $z^*$ and $z^{**}$ for $y^*$ and $y^{**}$. Then, by the definition of $H$,
$$H(y^*) = \sum_{S\subseteq [n]} \big( \sum_{i\in S} d_i\big)^{q} z^*_{S}\;\text{ and }\;H(y^{**}) = \sum_{S\subseteq [n]} \big( \sum_{i\in S} d_i\big)^{q} z^{**}_{S}.$$
Observe, that $\lambda z^* + (1-\lambda) z^{**}$ is a feasible solution for $\lambda y^* + (1-\lambda) y^{**}$ (since all LP constraints are linear). Hence, $H(\lambda y^* + (1-\lambda) y^{**})$ is at most the LP cost of $\lambda z^* + (1-\lambda) z^{**}$, which equals
$$\sum_{S\subseteq [n]} \Big( \sum_{i\in S} d_i\Big)^{q} (\lambda z^*_{S} + (1-\lambda) z^{**}_{S})
= \lambda H(y^*) + (1-\lambda) H(y^{**}).$$
\end{proof}

\section{Discretization}\label{sec:discr}
In this section, we show how to discretize values $d_{ij}$. We assume that functions $f_j$ satisfy the following  
conditions:
\begin{enumerate}
\item All $f_j$ are convex increasing nonnegative functions computable in polynomial time.
\item For every $j$, $f_j(0)=0$.
\item For every $j$ and $t\in \big[0,\sum_{i} d_{ij}\big]$,
$t(\log f_j(t))' \leq P$.
\end{enumerate}

We first find an approximate value of the optimal solution $\widetilde{OPT}\geq OPT$ and
then apply Theorem~\ref{thm:discr} (see below). Note that if the gap $\widetilde{OPT}/OPT$ is polynomially
bounded, then we can pick $\varepsilon$ such that the optimal value $OPT'$ of the discretized problem
is at most $(1+\varepsilon')OPT$. We pick such $\widetilde{OPT}$ either by using the binary search or
by enumerating all powers of 2 in the range $[\min_{ij} f_j(d_{ij}), \sum_j f_j(\sum_i d_{ij})]$.

\begin{theorem}\label{thm:discr}
There exists a polynomial-time algorithm that given an instance of the integer program
(\ref{IP:obj})-(\ref{IP:int}) satisfying conditions (1) and (2) above,
an upper bound $\widetilde{OPT}$ on the cost of the optimal solution $OPT$,
and $\varepsilon > 0$, returns a new set of coefficients $d'_{ij}$, numbers $\delta_j > 0$, and
an extra set of constraints $y_i=0$ for $i\in \calI$ such that each $d'_{ij}$ is a multiple of $\delta_j$; $d'_{ij}/\delta_j$
is an integer polynomially bounded in $n$, $1/\varepsilon$ and $P$ (see item 3 above)
such that the following two properties are satisfied.
\begin{enumerate}
\item The cost of the optimal solution for the new problem is at most the cost of the original problem:
$$OPT'\leq OPT.$$
\item For every feasible solution of the new problem $y\in \calP\cap \{0,1\}^n$ satisfying $y_i=0$ for $i\in \calI$,
we have
$$\sum_{j\in [k]} f_j \Big(\sum_{i\in [n]} d_{ij} y_{i}\Big)\leq (1+\varepsilon) \sum_{j\in [k]} f_j \Big(\sum_{i\in [n]} d'_{ij} y_{i}\Big) + \varepsilon \widetilde{OPT}.$$
Particularly,
$$OPT\leq (1+\varepsilon) OPT' + \varepsilon \widetilde{OPT}\leq OPT'+ 2\varepsilon \widetilde{OPT}.$$
\end{enumerate}
\end{theorem}
\begin{proof}
The proof is fairly standard: We round all $d_{ij}$ to be multiples of $\delta_j$. Then we show that if $\delta_j$'s are sufficiently small, then
the introduced rounding error is at most $2\varepsilon \widetilde{OPT}$. The details are below.

We algorithm finds the set $\calI = \{i: f_j(d_{ij})> \widetilde{OPT} \text{ for some } j\}$. If
$i\in \calI$, then $y_i$ must be equal to $0$ in every optimal solution, because otherwise,
$OPT\geq f_j(d_{ij})> \widetilde{OPT}$. Thus, for all $i\in \calI$, we set $y_i$ and all $d'_{ij}$'s
to be $0$. Then, we let
$$\eta = \min \big(\frac{\varepsilon}{4P},\frac{1}{2}\big);\;\;\;t_j = \max_{i\notin \calI} d_{ij};\;\;\;
\delta_j = \frac{\varepsilon \eta t_j}{kn}.$$
We round down all $d_{ij}$'s to be multiples of $\delta_j$. Denote the rounded values by $d'_{ij}$. This is our
new instance.

It is clear that $d'_{ij}$ are multiples of $\delta_j$. Since $t_j\geq d'_{ij}$ for all $i$ and $j$, we have  $d'_{ij}/\delta_j \le \frac{kn}{\varepsilon \eta }$
and therefore, all $d'_{ij}/\delta_j$ are polynomially bounded.  Since $d'_{ij} \leq d_{ij}$ and $f_j$ are monotone functions, we have for every $y$,
$$
\sum_{j\in [k]} f_j \Big(\sum_{i\in [n]} d'_{ij} y_{i}\Big)
\leq
\sum_{j\in [k]} f_j \Big(\sum_{i\in [n]} d_{ij} y_{i}\Big).
$$
As we observed earlier if $y^*$ is the optimal solution to the original problem, then $y^*_i=0$ for $i\in \calI$, hence $y^*$ is a feasible
solution to the new problem. Consequently,
$$OPT'\leq OPT.$$

We now need to verify that
$$\sum_{j\in [k]} f_j \Big(\sum_{i\in [n]} d_{ij} y_{i}\Big) - \sum_{j\in [k]} f_j \Big(\sum_{i\in [n]} d'_{ij} y_{i}\Big)
\leq \varepsilon \widetilde{OPT} + \varepsilon\sum_{j\in [k]} f_j \Big(\sum_{i\in [n]} d'_{ij} y_{i}\Big).$$
We prove that for every $j$,
\begin{equation}\label{eq:fj-diff}
f_j \Big(\sum_{i\in [n]} d_{ij} y_{i}\Big) - f_j \Big(\sum_{i\in [n]} d'_{ij} y_{i}\Big) \leq \frac{\varepsilon}{k}\, \widetilde{OPT} + \varepsilon f_j \Big(\sum_{i\in [n]} d'_{ij} y_{i}\Big).
\end{equation}
Consider two cases.

\medskip

\noindent I. If $\sum_{i\in [n]} d_{ij} y_{i}< \varepsilon t_j/k$, then
$$f_j \Big(\sum_{i\in [n]} d_{ij} y_{i}\Big)\leq \varepsilon f_j(t_j)/k \leq \varepsilon \widetilde{OPT}/k,$$
since $f$ is a convex function, $f(0)=0$, and $f_j(t_j)= \max_{i\notin \calI} f_j(d_{ij})\leq \widetilde{OPT}$. Hence, inequality (\ref{eq:fj-diff}) holds.

\medskip

\noindent II. Now assume that $\sum_{i\in [n]} d_{ij} y_{i}\geq \varepsilon t_j/k$. Observe, that
$$\sum_{i\in [n]} d_{ij} y_{i} - \sum_{i\in [n]} d'_{ij} y_{i}\leq n\delta_j
= \varepsilon \eta t_j/k \leq \eta \sum_{i\in [n]} d_{ij} y_{i}.$$
Hence (using that $\eta \in [0,1/2]$ and thus $1/(1-\eta)\leq 1 +2\eta$),
$$\sum_{i\in [n]} d_{ij} y_{i}\leq \frac{1}{1-\eta} \sum_{i\in [n]} d'_{ij} y_{i}
\leq (1+2\eta)\sum_{i\in [n]} d'_{ij} y_{i}.$$
We now use Claim~\ref{cl:f-dir} (see below) with 
$$t=\sum_{i\in [n]} d'_{ij} y_{i};\;\;\;\text{and}\;\;\;\eta' =  \frac{\sum_{i\in [n]} d_{ij} y_{i}}{\sum_{i\in [n]} d'_{ij} y_{i}} - 1\leq 2\eta\leq \frac{\varepsilon}{2P}.$$
We get
$$
f_j \Big(\sum_{i\in [n]} d_{ij} y_{i}\Big) - f_j \Big(\sum_{i\in [n]} d'_{ij} y_{i}\Big) \leq
\varepsilon f_j \Big(\sum_{i\in [n]} d'_{ij} y_{i}\Big).
$$
This finishes the proof. It only remains to prove Claim~\ref{cl:f-dir}.
\end{proof}

\begin{claim}\label{cl:f-dir}
Suppose that $f$ is a nonnegative monotonically increasing function such that for every $t\in(0,T]$, $t(\log f(t))'\leq P$. Let $\varepsilon \in [0,1]$
and $\eta' \leq \varepsilon/(2P)$.
Then, for $t\in (0, T/(1+\eta'))$,
$$f((1+\eta' )t) - f(t) \leq \varepsilon f(t).$$
\end{claim}
\begin{proof}
Write,
\begin{align*}
\log f((1+\eta' )t) - \log f(t) &= \int_t^{(1+\eta')t} (\log f(s))' ds
\leq \int_t^{(1+\eta')t} \frac{P}{s} ds\\&= P \log\frac{(1+\eta')t}{t} = P\log (1+\eta') \leq P\eta'.
\end{align*}
Thus,
$$\frac{f((1+\eta' )t)}{f(t)}\leq e^{\eta' P} \leq e^{\varepsilon/2},$$
and (since $e^{\varepsilon/2} - 1 \leq \varepsilon$ for $\varepsilon\in[0,1]$)
$$f((1+\eta' )t) - f(t) \leq (e^{\varepsilon/2} - 1) f(t)\leq \varepsilon f(t).$$
\end{proof}

\section{Lemma~\ref{lem:limsup}}
\begin{lemma}\label{lem:limsup}
Suppose that a sequence of integer random variables $S^n$ converges in distribution to an integer random variable $Z$.
Then, for every nonnegative function $\varphi$ with a finite expectation $\Exp[\varphi(Z)]$, we have
$$\sup_{n} \Exp[\varphi(S^n)]\geq \Exp[\varphi(Z)].$$
\end{lemma}
\begin{proof}
Let $\varphi_t (x) = \varphi(x)$ for $x\leq t$, and $\varphi_t(x)=0$, otherwise. The function $\varphi$ is nonzero on finitely many integral points.
Hence, $\lim_{n\to \infty} \Exp[\varphi_t(S^n)] = \Exp[\varphi_t(Z)]$ for every fixed $t$. Observe that $\Exp[\varphi(S^n)]\geq \Exp[\varphi_t(S^n)]$,
since $\varphi(x)\geq \varphi_t(x)$ for all $x$. On the other hand, $\lim_{t\to \infty}\Exp[\varphi_t(Z)] = \Exp[\varphi(Z)]$. Hence,
\begin{multline*}
\sup_{n} \Exp[\varphi(S^n)]\geq \sup_{n} \sup_{t} \Exp[\varphi_t(S^n)] =  \sup_{t} \sup_{n} \Exp[\varphi_t(S^n)]\geq\\
\sup_{t} \lim_{n\to\infty} \Exp[\varphi_t(S^n)] = \sup_t \Exp[\varphi_t(Z)]=\Exp[\varphi(Z)].
\end{multline*}
\end{proof}

\section{Convex Order}
For completeness, we give proofs of Lemma~\ref{lem:sumX-less-sumY} and Lemma~\ref{lem:sum-aX-less-aY} in this section. The reader may also
find slightly different proofs of these lemmas in the book of Shaked and Shanthikumar~\cite{SS-Book}.

\begin{lemma}\label{lem:XYZ}
Consider three independent random variables $X$, $Y$, and $Z$. If $X\cx Y$, then $Z+X\cx Z + Y$.
\end{lemma}
\begin{proof}
For every convex function $\varphi:\bbR\to \bbR$, we have
$$\E{\varphi (X+Z)} = \Exp \Exp[\varphi (X+Z)\given Z] \leq
\Exp \Exp[\varphi (Y + Z)\given Z] = \E{\varphi (Y + Z)}.$$
The inequality above holds, since for any fixed $Z=z$, the function $x\mapsto \varphi (x+z)$ is
convex.
\end{proof}
\begin{proof}[Proof of Lemma~\ref{lem:sumX-less-sumY}]
Using Lemma~\ref{lem:XYZ}, we replace $X_i$'s with $Y_i$'s in the sum $X_1+\dots+X_n$ one by one: For every $k$, let
$$
S_k = \sum_{i=1}^k X_i +  \sum_{j=k+1}^n Y_i.
$$
Then, by Lemma~\ref{lem:XYZ},
$$S_k = \Big(\sum_{i=1}^{k-1} X_i +  \sum_{j=k+1}^n Y_j\Big) + X_k \cx \Big(\sum_{i=1}^{k-1} X_i +  \sum_{j=k+1}^n Y_j\Big) + Y_k = S_{k-1}.$$
We have
$$\sum_{i=1}^n X_i = S_n \cx S_{n-1}\cx\dots \cx S_1 =\sum_{i=1}^n Y_i.$$
This finishes the proof.
\end{proof}

\begin{proof}[Proof of Lemma~\ref{lem:sum-aX-less-aY}]
Consider an arbitrary convex function $\varphi$. Let $A = \sum_{i=1}^n a_i$ and $\tilde{\varphi}(x) = \varphi(Ax)$. Since $\tilde{\varphi}$ is a convex function,
we have
$$\tilde{\varphi}\Big(\frac{1}{A}\sum_{i=1}^n a_i X_i\Big) \leq \frac{1}{A}\sum_{i=1}^n a_i \tilde{\varphi}(X_i),$$
and $\E{\tilde{\varphi}(X_i)}\leq \E{\tilde{\varphi}(Y)}$. Hence,
$$\E{\tilde{\varphi}\Big(\frac{1}{A}\sum_{i=1}^n a_i X_i\Big)} \leq \E{\frac{1}{A}\sum_{i=1}^n a_i \tilde{\varphi}(X_i)} = \frac{1}{A}\sum_{i=1}^n a_i \E{\tilde{\varphi}(X_i)}
\cx \frac{1}{A}\sum_{i=1}^n a_i \E{\tilde{\varphi}(Y)} = \E{\tilde{\varphi}(Y)}.$$
Substituting $\tilde{\varphi} (x) = \varphi(Ax)$, we get
$$\E{\varphi\Big(\sum_{i=1}^n a_i X_i\Big)} \leq \E{\tilde{\varphi}\Big(\sum_{i=1}^n a_i Y\Big)}.$$
\end{proof}

\section{Continuous Random Variables}\label{sec:crv}
In Section~\ref{sec:thm:decoupling}, we proved Theorem~\ref{thm:CX-main} for discrete random variables. In this section, we extend this result to
arbitrary random variables. Consider two sequences of nonnegative random variables $X_1,\dots, X_n$ and $Y_1,\dots, Y_n$ satisfying
conditions of Theorem~\ref{thm:CX-main}. Fix an arbitrary convex function $\varphi:\bbR\to\bbR$. We need to show
that
\begin{equation}\label{eq:crv:phi}
\Exp[\varphi(X_1+\dots+X_n)] \leq \Exp[\varphi(P\cdot (Y_1+\dots +Y_n))]
\end{equation}
assuming that both expectations exist. Observe that $\varphi$ can be represented as the sum of two convex functions: a monotonically non-decreasing
convex function $\fup$ and monotonically non-increasing convex function $\fdown$. It suffices to show that
\begin{align}
\Exp[\fup  (X_1+\dots+X_n)] &\leq \Exp[\fup  (P\cdot (Y_1+\dots +Y_n))];\label{eq:crv-phiUp}\\
\Exp[\fdown(X_1+\dots+X_n)] &\leq \Exp[\fdown(P\cdot (Y_1+\dots +Y_n))].\label{eq:crv-phiDown}
\end{align}
Note that if the expectations in (\ref{eq:crv:phi}) exist than the expectations in (\ref{eq:crv-phiUp}) and (\ref{eq:crv-phiDown})
also exist. We now prove inequality~(\ref{eq:crv-phiUp}). The proof of (\ref{eq:crv-phiDown}) is almost the same.
For natural $M$, define a function $g_M$ as follows:
$$g_M(x) =
\begin{cases}
\frac{\rounddown{M x}}{M},&\text{if } x \leq M;\\
0,&\text{otherwise}.
\end{cases}
$$
The function $g_M$ truncates $x$ at the level $M$ and then rounds $x$ down to the nearest multiple of $1/M$.
Observe, that $\lim_{M\to \infty} g_M(x) = x$ for every $x$. Hence, $(g_{M}(X_1) +\dots + g_{M}(X_n))$
converges a.s. to $X_1+\dots +X_n$ as $M\to\infty$; and $P(g_{M}(Y_1)+\dots +g_{M}(Y_n))$ converges a.s.
to $P(Y_1+\dots +Y_n)$ as $M\to \infty$. Since, $\fup$ is a continuous function
$\fup(g_{M}(X_1)+\dots +g_{M}(X_n))\xrightarrow{a.s.} \fup(X_1+\dots+X_n)$ and
$\fup(P(g_{M}(Y_1)+\dots +g_{M}(Y_n)))\xrightarrow{a.s.} \fup(P(Y_1+\dots+Y_n))$.
Notice that $g_M(x)\leq x$ for all $x$. Hence,
$\fup(g_{M}(X_1)+\dots +g_{M}(X_n))\leq \fup(X_1+\dots+X_n)$
and
$\fup(P(g_{M}(Y_1)+\dots +g_{M}(Y_n)))\leq \fup(P(Y_1+\dots+Y_n))$.
By Lebesgue's dominated convergence theorem, we get
\begin{align}
\lim_{M\to \infty} \Exp[\fup(g_{M}(X_1)+\dots +g_{M}(X_n))]&=\Exp[\fup(X_1+\dots+X_n)];\label{eq:coverg1}\\
\lim_{M\to \infty} \Exp[\fup(P(g_{M}(Y_1)+\dots +g_{M}(Y_n)))] &=  \Exp[\fup(P(Y_1+\dots+Y_n))].\label{eq:coverg2}
\end{align}
For every fixed $M$, the left hand side of (\ref{eq:coverg1}) is upper bounded by the left hand side of
(\ref{eq:coverg2}), because random variables $g_{M}(X_1),\dots, g_{M}(X_n)$ and $g_{M}(Y_1),\dots,g_{M}(Y_n)$ are
discrete and satisfy the conditions of Theorem~\ref{thm:CX-main}. Hence, the right hand side of (\ref{eq:coverg1})
is upper bounded by the right hand side of (\ref{eq:coverg2}). This proves inequality~(\ref{eq:crv-phiUp})
and concludes the proof of Theorem~\ref{thm:CX-main} for arbitrary random variables.

\section{Figures}\label{sec:figures}

\begin{figure}[ht]
\center{\begin{tabular}{|c|c|c|c|c|c|c|c|c|c|}
\hline
$q=$   & 1 &  1.25  &  1.5   &  1.75  & 2 \\ \hline
$A_q=$ & 1 &  1.163 &  1.373 &  1.645 & 2\\
\hline
\end{tabular}}
\caption{\label{fig:values} The values of $A_q$ for some $q\in [1,2]$. All values are rounded up to three decimal places.}
\end{figure}

\begin{figure*}[p]
\center{\includegraphics[width=350pt]{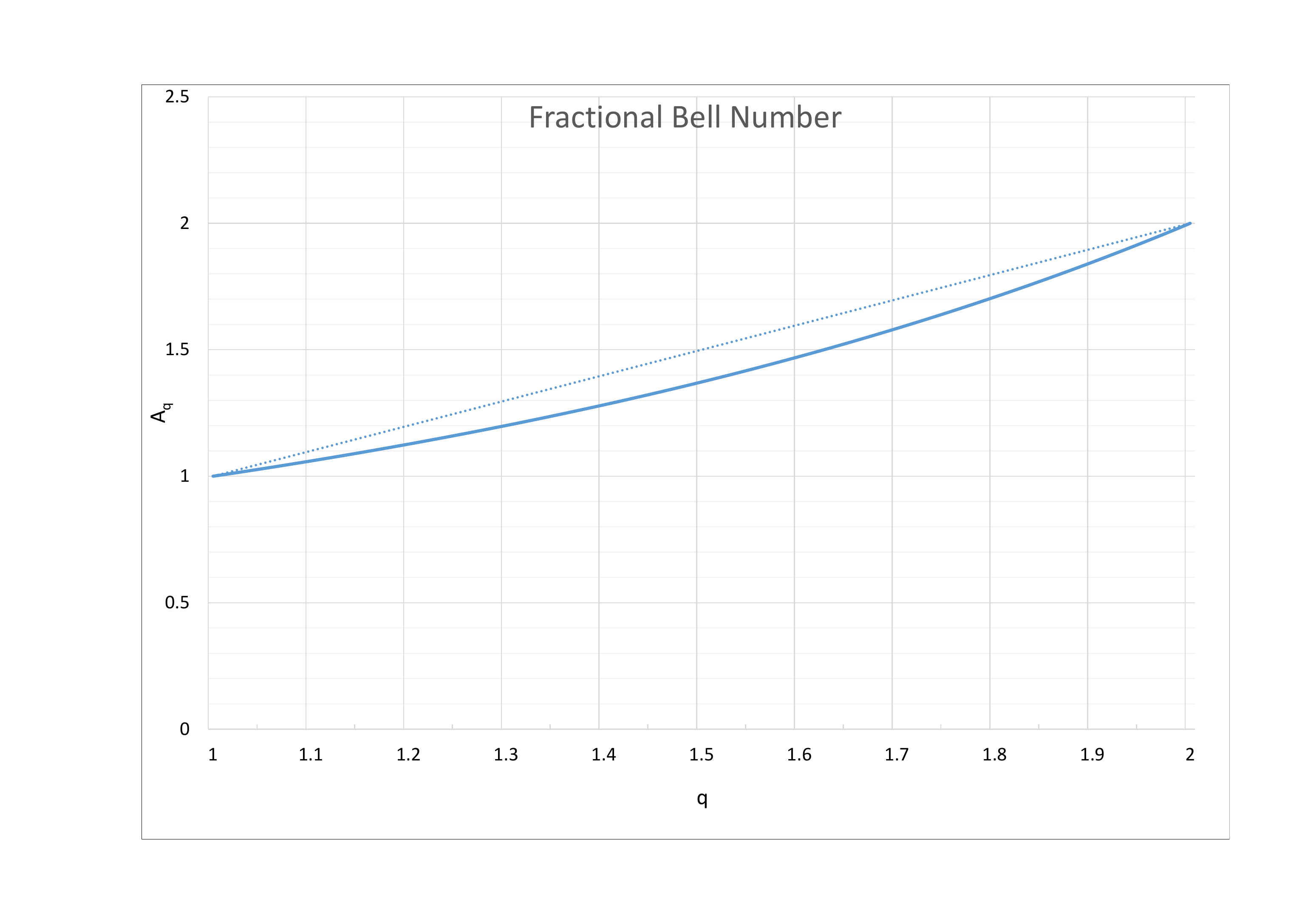}}
\caption{\label{fig:plot-Aq} Graph of $A_q$ for $q\in [1,2]$. Note that the function $q\mapsto A_q$ is convex; $A_1 = 1$ and $A_2 =2$. Thus, $A_q\leq q$ for $q\in [1,2]$.}
\end{figure*}

\begin{figure*}[p]
\center{\includegraphics[width=350pt]{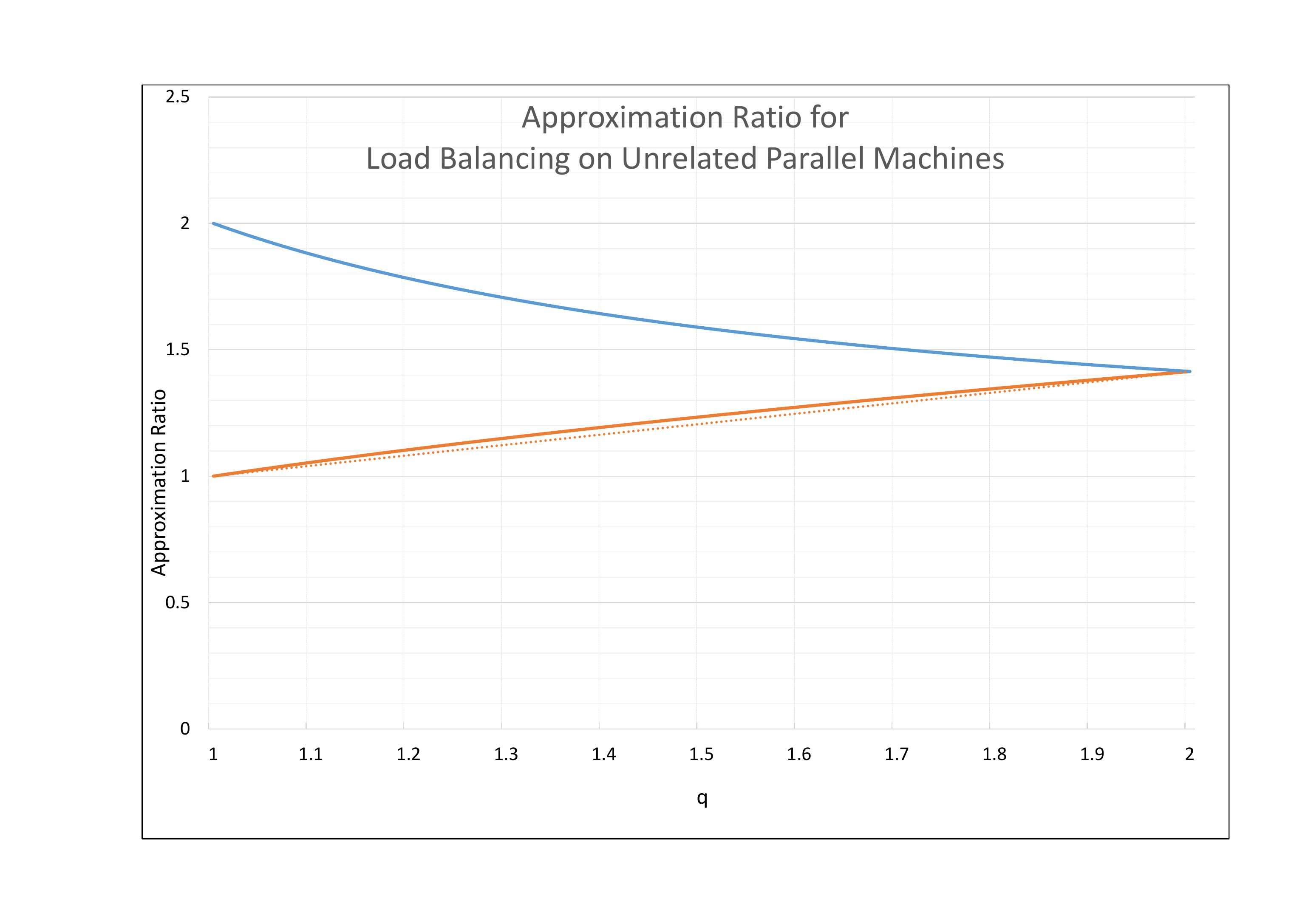}}
\caption{\label{fig:plot-results} Approximation factor of our algorithm for Load Balancing on Unrelated Parallel Machines -- $A_q^{\nicefrac{1}{q}}$ --
is plotted in red (below). Approximation factor of the algorithm due to Kumar, Marathe, Parthasarathy and Srinivasan~\cite{KMPS} -- $2^{\nicefrac{1}{q}}$ -- is plotted in
blue (above). The function $A_q^{\nicefrac{1}{q}}$ can be well approximated by the linear function
$1 + (\sqrt{2}-1)(q-1)$ in the interval $q\in [1,2]$.}
\end{figure*}

\begin{thebibliography}{2}

\bibitem{AS}
Alexander A. Ageev and Maxim Sviridenko.
Pipage Rounding: A New Method of Constructing Algorithms with Proven Performance Guarantee.
J. Comb. Optim. 8(3): 307-328, (2004).

\bibitem{A}
S. Albers.
Energy-efficient algorithms.
Commun. ACM 53(5), 86-96 (2010).

\bibitem{AAZZ} Matthew Andrews, Antonio Fernandez Anta, Lisa Zhang, Wenbo Zhao.
Routing for Power Minimization in the Speed Scaling Model. IEEE/ACM Trans. Netw. 20(1): 285-294 (2012).

\bibitem{AAZ} Matthew Andrews, Spyridon Antonakopoulos, Lisa Zhang.
Minimum-Cost Network Design with (Dis)economies of Scale.
FOCS 2010, pp.~585-592.

\bibitem{tree1} A. Assad and W. Xu.
The quadratic minimum spanning tree problem.
Naval Research Logistics v39. (1992), pp.~399-417.

\bibitem{AE} Yossi Azar, Amir Epstein.
Convex programming for scheduling unrelated parallel machines.
STOC 2005, pp.~331-337.

\bibitem{CCPV} Gruia Calinescu, Chandra Chekuri, Martin Pal and Jan Vondrak.
Maximizing a submodular set function subject to a matroid constraint.
SIAM Journal on Computing 40:6 (2011), pp. 1740-1766.

\bibitem{chvatal} V. Chvatal.
On certain polytopes associated with graphs.
J. Combinatorial Theory Ser. B, 18 (1975), pp. 138-154.

\bibitem{BKLLS} Evripidis Bampis, Alexander Kononov, Dimitrios Letsios,
Giorgio Lucarelli and Maxim Sviridenko.
Energy Efficient Scheduling and Routing via Randomized Rounding. FSTTCS 2013, pp.~449-460.

\bibitem{BP} Nikhil Bansal and Kirk Pruhs. The Geometry of Scheduling. FOCS 2010, pp.~407-414.

\bibitem{BV} Stephen Boyd and Lieven Vandenberghe. Convex Optimization. Cambridge University Press, 2004.

\bibitem{DV} C. Durr and O. Vasquez. Order constraints for single machine scheduling with non-linear cost.
The 16th Workshop on Algorithm Engineering and Experiments (ALENEX), 2014.

\bibitem{GLS} M. Grotschel, L. Lovasz and A. Schrijver. Geometric algorithms and combinatorial optimization. Springer-Verlag, Berlin, 1988.

\bibitem{Feller} W. Feller. An Introduction to Probability Theory and its Applications (3rd edition). John Wiley \& Sons, New York (1968).

\bibitem{fulkerson} D. R. Fulkerson. Blocking and anti-blocking pairs of polyhedra. Math. Programming, 1 (1971),
pp. 168-194.


\bibitem{G} Michel X. Goemans. Minimum Bounded Degree Spanning Trees. FOCS 2006, 273-282.


\bibitem{HJ} W. Hohn and T. Jacobs.
An experimental and analytical study of order constraints for single machine scheduling with quadratic cost.
14th Workshop on Algorithm Engineering and Experiments (ALENEX), pp. 103-117, SIAM, 2012.

\bibitem{HJ1} W. Hohn and T. Jacobs.
On the performance of Smith's rule in single-machine scheduling with nonlinear cost.
LATIN 2012, pp.~482-493.


\bibitem{HSW} Han Hoogeveen, Petra Schuurman, Gerhard J. Woeginger.
Non-Approximability Results for Scheduling Problems with Minsum Criteria.
INFORMS Journal on Computing 13(2), pp. 157-168 (2001).

\bibitem{IK}
S. Irani and K. Pruhs.
Algorithmic problems in power management.
ACM SIGACT News 36 (2), 63-76.

\bibitem{JP} Kumar Joag-Dev and Frank Proschan. Negative Association of Random Variables with Applications. The Annals of Statistics 11 (1983), no. 1, 286-295.

\bibitem{KMPS} V. S. A. Kumar, M. V. Marathe, S. Parthasarathy and A. Srinivasan. A Unified Approach to Scheduling on Unrelated Parallel Machines.
Journal of the ACM, Vol. 56, 2009.


\bibitem{tree2} S. Maia, E. Goldbarg, and M. Goldbarg. On the biobjective adjacent only quadratic spanning tree problem. Electronic Notes in Discrete Mathematics, 41, (2013), 535-542.

\bibitem{NN} Yurii Nesterov and Arkadi Nemirovski. Interior-point polynomial algorithms in convex programming. Society for Industrial and Applied Mathematics, 1987.


\bibitem{tree3} T. Oncan and A. Punnen. The quadratic minimum spanning tree problem: A lower bounding procedure and an efficient search algorithm. Computers and Operations Research v.37, (2010), 1762-1773.


\bibitem{Pena90} Victor de la Pe\~na. Bounds on the expectation of functions of martingales and sums of positive RVs in terms of norms of sums of independent random variables. Proceedings of the American Mathematical Society 108.1 (1990): 233-239.

\bibitem{Pena99} Victor de la Pe\~na and E. Gin\'e. Decoupling: from dependence to independence. Springer, 1999.

\bibitem{PIS} Victor de la Pe\~na, Rustam Ibragimov, Shaturgan Sharakhmetov. On Extremal
Distributions and Sharp $L_p$-Bounds for Sums of Multilinear Forms. The Annals of Probability,
2003, vol. 31 (2),pp.~630--675.

\bibitem{tree4} D. Pereira, M. Gendreau, A. Cunha. Stronger lower bounds for the quadratic minimum spanning tree problem with adjacency costs. Electronic Notes in Discrete Mathematics, 41, (2013), 229-236.


\bibitem{S98} A. Schrijver. Theory of Linear and Integer Programming. Wiley, 1998.

\bibitem{S02} A. Schrijver. Combinatorial Optimization. Springer, 2002.

\bibitem{SS-Book}
Moshe Shaked and J. George Shanthikumar. Stochastic orders. Springer, 2007.

\bibitem{Shao} Qi-Man Shao. A comparison theorem on moment inequalities between negatively associated and independent random variables.
Journal of Theoretical Probability 13, no. 2 (2000): 343-356.

\bibitem{SS} Andreas S. Schulz and Martin Skutella. Scheduling Unrelated Machines by Randomized Rounding. SIAM J. Discrete Math, 15(4), pp. 450-469 (2002).

\bibitem{SL} Mohit Singh and Lap Chi Lau. Approximating minimum bounded degree spanning trees to within one of optimal. STOC 2007: 661-670.

\bibitem{S} Martin Skutella. Convex quadratic and semidefinite programming relaxations in scheduling. J. ACM 48(2), 206-242 (2001).

\bibitem{S1} Martin Skutella. Approximation and randomization in scheduling.
Technische Universitat Berlin, Germany, 1998.

\bibitem{SW} S. Stiller and A. Wiese. Increasing Speed Scheduling and Flow Scheduling. ISAAC (2) 2010, pp.279-290.

\bibitem{power-ref} A. Wierman, L. L. H. Andrew, and A. Tang. Power-aware speed scaling in processor sharing systems. INFOCOM 2009, pp.~2007-2015.
\end{thebibliography}
\end{document}